\documentclass[letterpaper,11pt]{article}
\usepackage{times,mathptmx}
\DeclareMathAlphabet{\mathcal}{OMS}{cmsy}{m}{n}
\usepackage{fullpage}
\usepackage{amsfonts,amsmath,amssymb,amsthm,amscd,color}
\usepackage{tabularx,booktabs}
\usepackage{hyperref,cite,url}
\usepackage{array}
\hypersetup{pdfpagemode=UseNone,pdfstartview=}

\newcommand{\vs}{\vspace{2mm}}
\newtheorem{theorem}{Theorem}[section]
\newtheorem{assumption}[theorem]{Assumption}
\newtheorem{definition}[theorem]{Definition}

\newtheorem{lemma}[theorem]{Lemma}

\theoremstyle{definition}
\newtheorem{remark}[theorem]{Remark}

\newcommand{\G}{\mathbb{G}}
\newcommand{\Z}{\mathbb{Z}}
\newcommand{\bits}{\{0,1\}}
\newcommand{\Adv}{\textbf{Adv}}
\newcommand{\mc}[1]{\mathcal{#1}}
\newcommand{\tb}[1]{\textbf{#1}}
\newcommand{\mb}[1]{\mathbf{#1}}
\newcommand{\lb}{\linebreak[0]}

\title{Multi-Client Order-Revealing Encryption}

\author{
    Jieun Eom\footnote{Korea University, Seoul, Korea.
        Email: \texttt{jieunn.eom@gmail.com}.}
    \and
    Dong Hoon Lee\footnote{Korea University, Seoul, Korea.
        Email: \texttt{donghlee@korea.ac.kr}.}
    \and
    Kwangsu Lee\footnote{Sejong University, Seoul, Korea.
        Email: \texttt{kwangsu@sejong.ac.kr}.}
}

\date{}

\begin{document}

\maketitle

\begin{abstract}
Order-revealing encryption is a useful cryptographic primitive that
provides range queries on encrypted data since anyone can compare the order
of plaintexts by running a public comparison algorithm. Most studies on
order-revealing encryption focus only on comparing ciphertexts generated by
a single client, and there is no study on comparing ciphertexts generated
by multiple clients. In this paper, we propose the concept of multi-client
order-revealing encryption that supports comparisons not only on
ciphertexts generated by one client but also on ciphertexts generated by
multiple clients. We also define a simulation-based security model for
multi-client order-revealing encryption. The security model is defined with
respect to the leakage function which quantifies how much information is
leaked from the scheme. Next, we present two specific multi-client
order-revealing encryption schemes with different leakage functions in
bilinear maps and prove their security in the random oracle model. Finally,
we give the implementation of the proposed schemes and suggest methods to
improve the performance of ciphertext comparisons.
\end{abstract}

\vs \noindent {\bf Keywords:} Symmetric-key encryption, Order-revealing
encryption, Multi-client order-revealing encryption, Bilinear maps.

\newpage
\section{Introduction}

Today, a large amount of the users' data is collected and stored in cloud
servers to provide various services utilizing this personal data. Recently,
as the concern of privacy issues in personal data has increased, it has been
an important issue to safely store personal data in a cloud server and to
prevent it from being leaked. The simplest way to solve this issue is to
perform data encryption. However, it is difficult for the cloud server to
provide ordinary services such as keyword searches, range queries, and
numeric operations on encrypted data since plaintexts are transformed to
random ciphertexts. In order to overcome this problem, advanced encryption
schemes that support computation on encrypted data such as homomorphic
encryption and functional encryption have been actively studied
\cite{Gentry09,GargGHRSW13}. However, it is difficult to provide efficient
services using them since these schemes are somewhat inefficient.

One way to allow efficient computation on encrypted data while providing
privacy of user data is to consider an efficient encryption scheme that
allows only a limited operation such as a search or range query. Searchable
symmetric encryption (SSE) is a kind of symmetric-key encryption that
supports keyword searching on encrypted data \cite{SongWP00}.
Order-preserving encryption (OPE) and order-revealing encryption (ORE) are
special kinds of symmetric-key encryption that can be used for efficient
range queries over encrypted data by comparing ciphertexts without decrypting
these ciphertexts. An OPE scheme is a deterministic encryption scheme, which
encrypts plaintexts in numeric values to generate ciphertexts in numerical
values by maintaining the order of plaintexts, so that the order of
plaintexts can be compared by simply comparing the order of ciphertexts
\cite{AgrawalKSX04,BoldyrevaCLO09,BoldyrevaCO11}. An ORE scheme is a
probabilistic encryption scheme having ciphertexts of arbitrary values, and
the order of plaintexts can be compared by running a public comparison
algorithm on ciphertexts \cite{BonehLR+15,ChenetteLWW16,LewiW16,CashLOZ16}.
The first ORE scheme of Boneh et al. \cite{BonehLR+15} provides the best
possible security, but it is inefficient since it uses heavy cryptographic
tools such as multi-linear maps. Recently, several practical ORE schemes have
been proposed but these schemes inevitably leak some information on
plaintexts in addition to the comparison result \cite{ChenetteLWW16,LewiW16,
CashLOZ16}.

All of the previous ORE studies only considered to compare ciphertexts
generated by a single client. However, in a real environment, it is necessary
to compare ciphertexts generated by multiple clients if these clients handle
related plaintexts. For example, we consider a scenario where students are
divided into multiple classes to take lectures taught by different
instructors. In this case, the grades of each class are encrypted by the
encryption key of each instructor, but if necessary, the grades of these
different classes should be comparable without decryption. As another
example, we can consider a scenario in which patients are treated by
different physicians in a hospital and their medical data are encrypted and
stored with the secret keys of physicians. In this case, a physician may want
to compare the medical data of patients that he or she has treated with the
medical data of other patients that have been treated by other physicians for
medical research purposes. To support these scenarios, a comparison key must
be provided that can compare the encrypted data generated by multiple clients
and this comparison key should be provided only to an authorized user. We
call the ORE scheme that supports comparison operations not only on
ciphertexts generated by one client but also on ciphertexts generated by
different clients, as the multi-client order-revealing encryption (MC-ORE)
scheme.

We note that an MC-ORE scheme can be easily derived from a multi-input
functional encryption (MI-FE) scheme \cite{GoldwasserGG+14}. That is, if each
ciphertext slot of an MI-FE scheme is related to the client index of an
MC-ORE scheme and an MI-FE private key for the comparison function on two
ciphertexts is provided as an MC-ORE comparison key, then we can build an
MC-ORE scheme from an MI-FE scheme. However, this approach is not practical
because an MI-FE scheme for general functions requires heavy cryptographic
tools such as multi-linear maps or indistinguishable obfuscation.

\subsection{Our Results}

We summarize the contributions of this paper which include the notion of
MC-ORE and two practical MC-ORE schemes with limited leakage in bilinear
maps.

\vs\noindent \tb{Definition.} We first introduce the notion of MC-ORE by
extending the concept of ORE \cite{BonehLR+15} to additionally support the
comparison operation on ciphertexts which are generated by multiple clients.
In an MC-ORE scheme, each client creates ciphertexts by encrypting plaintexts
with his/her secret key and anyone can publicly compare the order of two
ciphertexts generated by a single client similar to the functionality of ORE.
In addition to this basic functionality, it supports the comparison operation
of two ciphertexts created by different clients if an additional comparison
key for two clients is given. Note that the comparison of two ciphertexts
generated from different clients is not a public operation since a comparison
key given from a trusted center is needed to prevent the leakage resulting
from these comparisons. To define the security model of MC-ORE, we follow the
security model of ORE that allows the leakage \cite{ChenetteLWW16}. In this
work, we give a simulation-based security model for MC-ORE with a leakage
function $\mc{L}$. Informally, this definition states that if an adversary
can obtain information from ciphertexts of clients' plaintexts $(j_1, m_1),
\ldots, (j_q, m_q)$ where $j_k$ is the index of a client, then it can be
inferred from $\mc{L}((j_1, m_1), \ldots, (j_q, m_q))$. One difference
between our security model and that of ORE with the leakage is that the
adversary can query many comparison keys for different clients. To handle
this comparison key query, we define the static security model which requires
that the adversary should first specify a set of corrupted client indices.

\vs\noindent \tb{Basic Construction.} Next we propose two realizable MC-ORE
schemes with different leakage functions. Our first MC-ORE scheme
conceptually follows the design principle of the ORE scheme of Chenette et
al. \cite{ChenetteLWW16} that encrypts each bit of a plaintext by using a
pseudo-random function (PRF) that takes a prefix of the plaintext as an
input. However, it is not easy to extend an ORE scheme that uses a PRF to an
MC-ORE scheme that supports the comparison operation for different clients
since the outputs of PRF with different client's keys are random values. To
solve this difficulty, we use an algebraic PRF in bilinear groups which is
defined as $PRF_s(x) = H(x)^s$ where $H$ is a hash function and $s$ is a PRF
key \cite{NaorPR99}. Suppose there is a single client and the client creates
ciphertexts $C = (H(x)^s, H(x+1)^s)$ and $C' = (H(x')^s, H(x'+1)^s)$ for
plaintexts $x$ and $x'$ in binary values by using a secret key $s$. A user
can publicly check whether $x+1 = x'$ or not by comparing $H(x+1)^s =
H(x')^s$ from two ciphertexts. Now suppose there are two clients with
different secret keys $s$ and $s'$ and clients create ciphertexts $C =
(H(x)^s, H(x+1)^s)$ and $C' = (H(x')^{s'}, H(x'+1)^{s'})$ for plaintexts $x$
and $x'$ in binary values respectively. To compare two ciphertexts generated
by different clients, a user first receives a comparison key $CK =
(\hat{g}^{rs}, \hat{g}^{rs'})$ from a trusted center and checks whether $x+1
= x'$ or not by comparing $e(H(x+1)^s, \hat{g}^{rs'}) = e(H(x')^{s'},
\hat{g}^{rs})$. To extend the comparison of binary values to large values, we
modify the encoding method of Chenette et al. \cite{ChenetteLWW16} that uses
the prefixes of a plaintext. Let $m = x_1 x_2 \cdots x_n \in \bits^n$ be a
plaintext. For each $i \in [n]$, the encryption algorithm encodes two strings
$E_{i,0} = x_1 x_2 \cdots x_{i-1} \| 0 x_i$ and $E_{i,1} = x_1 x_2 \cdots
x_{i-1} \| (0 x_i + 1)$ and evaluates $C_{i,0} = H(E_{i,0})^s$ and $C_{i,1} =
H(E_{i,1})^s$. For example, the third bit of $m = 101$ is encoded as $E_{3,0}
= 10 \| 01 = 1001$ and $E_{3,1} = 10 \| (01 + 1) = 10 \| 10 = 1010$. The
ciphertext is formed as $CT = (\{ C_{i,0}, C_{i,1} \}_{i \in [n]})$. Note
that we have $m < m'$ if there is the smallest index $i^*$ such that the
prefixes of two plaintexts with $i^*-1$ length are equal and $x_{i^*} + 1 =
x'_{i^*}$. We prove the security of our first MC-ORE scheme in the
simulation-based (SIM) security model with a leakage function that reveals
the comparison result as well as the most significant differing bit.

\vs\noindent \tb{Enhanced Construction.} Our second MC-ORE scheme is the
enhanced version of the first MC-ORE scheme that reduces the leakage due to
the comparison of ciphertexts generated by a single client. In our first
scheme, a ciphertext was simultaneously used for two purposes: ciphertext
comparisons in a single client and ciphertext comparisons between different
clients. In the second scheme, we divide the ciphertext into two parts and
treat each ciphertext part differently. That is, the first ciphertext part is
only used for ciphertext comparisons in a single client, and the second
ciphertext part is only used for ciphertext comparisons between different
clients. For the first ciphertext part, we can use any ORE scheme that has
the reduced leakage \cite{CashLOZ16,LewiW16}. For the second ciphertext part,
we construct an encrypted ORE (EORE) scheme by modifying our first MC-ORE
scheme. In the EORE scheme, an (encrypted) ciphertext is created by first
generating a ciphertext of the first MC-ORE scheme and then encrypting it
with a public-key encryption scheme. Unlike the first MC-ORE scheme, this
EORE scheme does not allow ciphertext comparisons in a single client since
ciphertexts are securely encrypted. However, it allows ciphertext comparisons
between different clients since it can derive the original ciphertexts of the
first MC-ORE scheme if a comparison key is provided by the trusted center.
Therefore, there is \textit{no leakage} from the second ciphertext part and
the leakage only depends on the first ciphertext part if comparison keys are
not exposed. We prove the SIM security of our second MC-ORE scheme under the
external Diffie-Hellman assumption.

\vs\noindent \tb{Implementation.} Finally, we implement our MC-ORE schemes
and evaluate the performance of each algorithm. The proposed MC-ORE scheme
provides single-client comparison and multi-client comparison algorithms. In
the MC-ORE scheme, the most computationally expensive algorithm is the
multi-client comparison algorithm since it requires two pairing operations
per each bit comparison until the most significant differing bit (MSDB) is
found. To improve this multi-client comparison, we present other comparison
methods and compare the performance of these suggested methods. The first
method is a simple method that performs the comparison sequentially from the
ciphertext element of the most significant bit to that of the least
significant bit. It is efficient when the MSDB exists in the higher bits, but
it is inefficient when the MSDB exists in the lower bits. The second method
is a binary search method that uses a binary search instead of a sequential
search to find the MSDB location. This method performs approximately $\log n$
computations to find the MSDB location where $n$ is the length of a
plaintext. The third method is a hybrid method that combines multi-client
comparisons and single-client comparisons. This method can improve the
performance of ciphertext comparisons between multiple ciphertexts by
performing one multi-client comparison and many single-client comparisons.

\subsection{Related Work}

\noindent \tb{Order-Preserving Encryption.} The concept of OPE was introduced
by Agrawal et al. \cite{AgrawalKSX04} in the database community, and this is
a symmetric-key encryption scheme that supports efficient comparison
operations on ciphertexts since the order of plaintexts is maintained in
ciphertexts. The security model of OPE was presented by Boldyreva et al.
\cite{BoldyrevaCLO09}, and it is called indistinguishability under ordered
chosen plaintext attack (IND-OCPA). The security notion of IND-OCPA says that
an adversary can not obtain any information from ciphertexts except the order
of underlying plaintexts. However, the ciphertext space of OPE is required to
be extremely large to satisfy this IND-OCPA security. To achieve this
IND-OCPA security, several variants of OPE such as mutable OPE have been
proposed, but most of them are inefficient since they require stateful
encryption and an interactive protocol \cite{PopaLZ13,KerschbaumS14,
RocheACY16}.

\vs\noindent \tb{Order-Revealing Encryption.} Boneh et al. \cite{BonehLR+15}
introduced the notion of ORE which is a generalization of OPE where the order
of plaintexts can be publicly compared by running a comparison algorithm on
ciphertexts. They also proposed a specific ORE scheme that achieves the
IND-OCPA security by using multi-linear maps, but this scheme is quite
impractical. Chenette et al. \cite{ChenetteLWW16} constructed the first
practical ORE scheme by encrypting each bit of messages using pseudo-random
functions. They showed that their scheme achieve a weaker security model of
ORE that reveals additional information of underlying plaintexts in addition
to the order of plaintexts. After the work of Chenette et al., many ORE
schemes were proposed to reduce the additional leakage. Lewi et al.
\cite{LewiW16} constructed an IND-OCPA secure ORE scheme only for small
plaintext spaces by decomposing the encryption algorithm into two separate
functions, left encryption and right encryption where the right encryption
achieves the IND-OCPA security. Cash et al. \cite{CashLOZ16} constructed an
ORE scheme with reduced leakage by using property-preserving hash functions
in bilinear maps. Although this ORE scheme achieves to reduce the leakage, it
is inefficient due to the larger size of ciphertexts and the pairing
operation.

\vs\noindent \tb{Attacks on ORE.} Naveed et al. \cite{NaveedKW15} explored
inference attacks on encrypted database columns to recover messages against
ORE-encrypted databases. These attacks usually use the order and frequency of
plaintexts and auxiliary information such as plaintext distribution. Durak et
al. \cite{DurakDC16} and Grubbs et al. \cite{GrubbsSB+17} proposed improved
inference attacks of Naveed et al. in several ways and additionally presented
leakage-abuse attacks against ORE schemes with the specified leakage. Both
attacks showed that the leakage of ORE can be effectively used to recover
more accurate plaintexts than that was theoretically analyzed.

\section{Multi-Client Order-Revealing Encryption}

In this section, we define the syntax and the security model of multi-client
order-revealing encryption by extending those of order-revealing encryption.

\subsection{Notation}

Let $[n]$ be the set of $\{ 1, \ldots, n \}$ and $[k,n]$ be the set of $\{ k,
\ldots, n \}$. Let $\mb{cmp}(m, m')$ be a comparison function that returns 1
if $m < m'$ and returns 0, otherwise. Let $\mb{ind}(m, m')$ be an index
function that returns the index of the most significant differing bit between
plaintexts $m$ and $m'$ of $n$-bits and returns $n+1$ if $m = m'$. Let
$\mb{prefix}(m, i)$ be a prefix function that takes as input a plaintext $m =
x_1 x_2 \cdots x_n \in \bits^n$ and an index $i$ and returns $x_1 x_2 \cdots
x_{i-1}$ as the prefix of $x_i$.

\subsection{Order-Revealing Encryption}

Order-revealing encryption (ORE) is a special kind of symmetric-key
encryption that supports a comparison operation on encrypted data by using a
public procedure \cite{BonehLR+15}. In ORE, a client creates ciphertexts of
plaintexts by using his/her secret key $SK$ and uploads these ciphertexts to
a remote database. After that, anyone can compare the order of two
ciphertexts $CT$ and $CT'$ by using a public comparison algorithm. The
following is the syntax of ORE given by Chenette et al. \cite{ChenetteLWW16}.

\begin{definition}[ORE] \label{def:ore}
An ORE scheme consists of three algorithms, \tb{Setup}, \tb{Encrypt},
\tb{Compare} which are defined over a well-ordered domain $\mc{D}$ as
follows:
\begin{description}
\item [\tb{Setup}($1^\lambda$).] The setup algorithm takes as input a
    security parameter $\lambda$ and outputs a secret key $SK$.

\item [\tb{Encrypt}($m, SK$).] The encryption algorithm takes as input a
    plaintext $m \in \mc{D}$ and the secret key $SK$ and outputs a
    ciphertext $CT$.

\item [\tb{Compare}($CT, CT'$).] The comparison algorithm takes as input
    two ciphertexts $CT$ and $CT'$ and outputs a comparison bit $b \in
    \bits$.
\end{description}
The correctness of ORE is defined as follows: For all $SK$ generated by
\tb{Setup} and any $CT, CT'$ generated by \tb{Encrypt} on plaintexts $m, m'$,
it is required that
    $\mb{Compare}(CT, CT') = \mb{cmp}(m, m')$.
\end{definition}

The best possible security of ORE, which is IND-OCPA, was defined by Boneh et
al. \cite{BonehLR+15}. The simulation-based security of ORE with additional
leakage $\mc{L}$ was defined by Chenette et al. \cite{ChenetteLWW16}.

\subsection{Multi-Client Order-Revealing Encryption}

Multi-client order-revealing encryption (MC-ORE) is an extension of ORE that
supports comparison operations not only between ciphertexts generated by a
single client but also between ciphertexts generated by different clients. In
MC-ORE, each client of an index $j$ creates ciphertexts of plaintexts by
using his/her secret key $SK_j$ which is given by a trusted center. Anyone
can compare two ciphertexts $CT_j$ and $CT_j'$ generated by the single client
by using a public comparison algorithm as the same as in ORE. In addition, a
client can compare two ciphertexts $CT_j$ and $CT_k'$ generated by different
clients with different indices $j$ and $k$ if the client obtains a comparison
key $CK_{j,k}$ from the trusted center. The syntax of MC-ORE is given as
follows.

\begin{definition}[MC-ORE] \label{def:mc-ore}
An MC-ORE scheme consists of six algorithms, \tb{Setup}, \tb{GenKey},
\tb{Encrypt}, \tb{Compare}, \tb{GenCmpKey}, and \tb{CompareMC}, which are
defined as follows:
\begin{description}
\item [\tb{Setup}($1^\lambda, N$).] The setup algorithm takes as input a
    security parameter $\lambda$ and the number of clients $N \in
    \mathbb{N}$ and outputs a master key $MK$ and public parameters $PP$.

\item [\tb{GenKey}($j, MK, PP$).] The key generation algorithm takes as input
    a client index $j \in [N]$, the master key $MK$, and the public
    parameters $PP$. It outputs a secret key $SK_j$ for the client index
    $j$.

\item [\tb{Encrypt}($m, SK_j, PP$).] The encryption algorithm takes as input
    a plaintext $m \in \mc{D}$, the secret key $SK_j$, and the public
    parameters $PP$. It outputs a ciphertext $CT_j$.

\item [\tb{Compare}($CT_j, CT'_j, PP$).] The comparison algorithm takes as
    input two ciphertexts $CT_j, CT'_j$ of the same client index $j$ and
    the public parameters $PP$. It outputs a comparison bit $b \in \bits$.

\item [\tb{GenCmpKey}($j, k, MK, PP$).] The comparison key generation
    algorithm takes as input two client indices $j, k$, the master key
    $MK$, and the public parameters $PP$. It outputs a comparison key
    $CK_{j,k}$ for two clients.

\item [\tb{CompareMC}($CT_j, CT'_k, CK_{j,k}, PP$).] The multi-client
    comparison algorithm takes as input two ciphertexts $CT_j, CT'_k$ of
    two client indices $j, k$, the comparison key $CK_{j,k}$, and the
    public parameters $PP$. It outputs a comparison bit $b \in \bits$.
\end{description}
The correctness of MC-ORE is defined as follows: For all $PP, MK, \{SK_j\}_{j
\in [N]}$ generated by \tb{Setup} and \tb{GenKey}, any $CK_{j,k}$ generated
by \tb{GenCmpKey}, and any $CT_j, CT'_j, CT''_k$ generated by \tb{Encrypt} on
plaintexts $m, m', m''$, it is required that:
    \begin{align*}
    & \mb{Compare}(CT_j, CT'_j, PP) = \mb{cmp}(m, m') \text{ and } \\
    & \mb{CompareMC}(CT_j, CT''_k, CK_{j,k}, PP) = \mb{cmp}(m, m'').
    \end{align*}
\end{definition}

The simulation-based security (SIM-security) model of MC-ORE is defined with
a leakage function which enables quantifying any information inevitably
leaked from the scheme. Since the leakage is affected by whether comparison
keys are exposed, the leakage function $\mc{L}_S$ is defined with respect to
a set $S$ of the revealed comparison keys. In the real experiment, an
adversary can access a comparison key generation oracle to obtain any
comparison key as well as an encryption oracle to obtain any ciphertext of
its choice $(j_i, m_i)$ where $j_{i}$ is the client index corresponding to
the $i$-th message $m_{i}$. Eventually, the adversary outputs the deducing
result from the given information. In the ideal experiment, the adversary
also can obtain any comparison key and any ciphertext, but all values are
generated by the simulator which has only the information derived from the
leakage function $\mc{L}_S((j_1, m_1), \ldots, (j_q, m_q))$. The security is
proved by showing the outputs of two distributions are indistinguishable.

However, the leakage function is influenced by the order of ciphertext
queries and comparison key queries. When $(j_1,m_1)$ and $(j_2,m_2)$ are
queried to the encryption oracle, the simulator generates ciphertexts
$CT_{j_1}$ and $CT_{j_2}$ with no leakage if the comparison key
$CK_{j_1,j_2}$ was not exposed. After that, if the adversary requests
$CK_{j_1,j_2}$ causing the leakage $\mc{L}_S((j_1,m_1), (j_2,m_2))$, it can
identify that there is something wrong in the simulation of $CT_{j_1}$ and
$CT_{j_2}$. That is, the simulator should have generated the ciphertexts by
predicting the leakage but it is difficult to simulate with such a flexible
leakage function.
In addition, when $CK_{j_1,j_3}$ and $CK_{j_2,j_3}$ are exposed, the
simulator generates $CT_{j_1}$ and $CT_{j_2}$ with no leakage since
$CK_{j_1,j_2}$ is not exposed. After that, if $(j_3,m_3)$ is queried to the
encryption oracle, the simulator generates $CT_{j_3}$ with the leakage
$\mc{L}_S((j_1,m_3), (j_2,m_3))$. Again, the adversary can notice that the
simulation of $CT_{j_1}$ and $CT_{j_2}$ is wrong. Thus, we define the static
version of the SIM-security model in which a set $S$ of revealed comparison
keys and the ciphertext queries are initially fixed. The static SIM-security
model of MC-ORE with the leakage function $\mc{L}_S$ is defined as follows.

\begin{definition}[\tb{Static SIM-Security with Leakage}] \label{def:static-security}
For a security parameter $\lambda$, let $\mc{A}$ be an adversary and $\mc{B}$
be a simulator. Let $S = \{(j,k)\}_{j,k\in [N]}$ be a set of index tuples
where $CK_{j,k}$ is revealed and let $\mc{L}_S(\cdot)$ be a leakage function.
The experiments of $\mathsf{REAL}_{\mc{A}}^{ MC\text{-}ORE}(\lambda)$ and
$\mathsf{SIM}_{\mc{A}, \mc{B}, \mc{L}}^{MC\text{-}ORE}(\lambda)$ are defined
as follows:

\begin{center}
\begin{tabular}{|c|} \hline
\begin{tabular}{l}
$\underline{\mathsf{REAL}_{\mc{A}}^{MC\text{-}ORE}(\lambda)}$\\
1. $\big(st_{\mc{A}}, S, ((j_1, m_1), \cdots, (j_q, m_q))\big)\gets \mc{A}(1^\lambda)$\\
2. $(PP, MK) \gets \mb{Setup}(1^\lambda, N)$\\
3. $CK_{j,k} \gets \mb{GenCmpKey}(j,k,MK,PP), ~\forall (j,k) \in S$\\
4. for $1\le i\le q$,\\
~~~~$CT_{j_i} \gets \mb{Encrypt}(m_i, SK_{j_i}, PP)$\\
5. Output $(CT_{j_1},\cdots, CT_{j_q})$ and $st_{\mc{A}}$\\
\end{tabular}
\\\\
\begin{tabular}{l}
$\underline{\mathsf{SIM}_{\mc{A},\mc{B},\mc{L}}^{MC\text{-}ORE}(\lambda)}$\\
1. $\big(st_{\mc{A}}, S, ((j_1, m_1), \cdots, (j_q, m_q))\big)\gets \mc{A}(1^\lambda)$\\
2. $(st_{\mc{B}}, PP) \gets \mc{B}(1^\lambda, N)$\\
3. $CK_{j,k} \gets \mc{B}(st_{\mc{B}}), ~\forall (j,k)\in S$\\
4. for $1\le i\le q$,\\
~~~~$\big( st_{\mc{B}}, CT_{j_i} \big) \gets \mc{B} \big( st_{\mc{B}},
\mc{L}_S((j_1, m_1), \cdots, (j_i, m_i)) \big)$\\
5. Output $(CT_{j_1},\cdots, CT_{j_q})$ and $st_{\mc{A}}$\\
\end{tabular}
\\\hline
\end{tabular}
\end{center}

\vs \noindent We say that an MC-ORE scheme is ST-SIM secure if for all
polynomial-size adversaries $\mc{A}$, there exists a polynomial-size
simulator $\mc{B}$ such that the outputs of the two distributions
$\mathsf{REAL}_{\mc{A}}^{MC\text{-}ORE}(\lambda)$ and $\mathsf{SIM}_{\mc{A},
\mc{B}, \mc{L}}^{MC\text{-}ORE}(\lambda)$ are indistinguishable.
\end{definition}

\begin{remark}
For $S = \{(j,k)\}_{j,k\in [N]}$ of index tuples where the comparison key
$CK_{j,k}$ is revealed, let $\mc{L}_S$ be the following leakage function:
    \begin{align*}
    &\mc{L}_S \big( (j_1, m_1), \cdots, (j_q, m_q) \big)
    = \big\{ \mb{cmp}(m_{i'}, m_{i}):
      1 \le i' < i \le q, j_{i'} = j_i ~\text{or}~ (j_{i'}, j_{i}) \in S \big\}.
    \end{align*}
If an MC-ORE scheme is secure with leakage $\mc{L}_S$, then it is IND-OCPA
secure.
\end{remark}

\section{Basic MC-ORE Construction}

In this section, we propose our first construction of MC-ORE with leakage and
prove the ST-SIM security of our scheme.

\subsection{Asymmetric Bilinear Groups}

Let $\mc{G}_{as}$ be a group generator algorithm that takes as input a
security parameter $\lambda$ and outputs a tuple $(p, \G, \hat{\G}, \G_T, e)$
where $p$ is a random prime and $\G, \hat{\G}$, and $\G_T$ be three cyclic
groups of prime order $p$. Let $g$ and $\hat{g}$ be generators of $\G$ and
$\hat{\G}$, respectively. The bilinear map $e : \G \times \hat{\G}
\rightarrow \G_{T}$ has the following properties:
\begin{enumerate}
\item Bilinearity: $\forall u \in \G, \forall \hat{v} \in \hat{\G}$ and
    $\forall a,b \in \Z_p$, $e(u^a,\hat{v}^b) = e(u,\hat{v})^{ab}$.
\item Non-degeneracy: $\exists g \in \G, \hat{g} \in \hat{\G}$ such that
    $e(g,\hat{g})$ has order $p$ in $\G_T$.
\end{enumerate}
We say that $\G, \hat{\G}, \G_T$ are asymmetric bilinear groups if the group
operations in $\G, \hat{\G}$, and $\G_T$ as well as the bilinear map $e$ are
all efficiently computable, but there are no efficiently computable
isomorphisms between $\G$ and $\hat{\G}$.

\begin{assumption}[External Diffie-Hellman, XDH]
Let $(p, \G, \hat{\G}, \G_T, e)$ be a tuple randomly generated by
$\mc{G}_{as}(1^\lambda)$ where $p$ is a prime order of the groups. Let $g,
\hat{g}$ be random generators of groups $\G, \hat{\G}$, respectively. The XDH
assumption is that the decisional Diffie-Hellman(DDH) assumption holds in
$\G$. That is, if the challenge tuple
    $$D = \big( (p, \G, \hat{\G}, \G_T, e), g, \hat{g}, g^a, g^b \big)
    \mbox{ and } T$$
are given, no PPT algorithm $\mc{A}$ can distinguish $T = T_0 = g^{ab}$ from
$T = T_1 = g^c$ with more than a negligible advantage. The advantage of
$\mc{A}$ is defined as
    $\Adv_{\mc{A}}^{XDH} (\lambda) = \big|
    \Pr[\mc{A}(D,T_0) = 0] - \Pr[\mc{A}(D,T_1) = 0] \big|$
where the probability is taken over random choices of $a, b, c \in \Z_p$.
\end{assumption}

\subsection{Construction}

Before we present our basic MC-ORE scheme, we first define a leakage function
for our scheme. Let $N \in \mathbb{N}$ be the maximum number of clients and
$S = \{(j,k)\}_{j,k\in [N]}$ be a set of client index tuples where a
comparison key $CK_{j,k}$ is revealed. A leakage function $\mc{L}_S$ is
defined as follows:
    \begin{align*}
    \mc{L}_S \big( (j_1, m_1), \cdots, (j_q, &m_q) \big) = \big\{ \mb{cmp}(m_{i'}, m_{i}),
    \mb{ind}(m_{i'}, m_{i}): 1 \le i' < i \le q, j_{i'} = j_i ~\text{or}~ (j_{i'}, j_{i}) \in S \big\}.
    \end{align*}

If $S = \emptyset$, then $\mc{L}_S$ becomes equal to the leakage function
$\mc{L}$ defined by Chenette et al.\cite{ChenetteLWW16}. Otherwise, i.e. if
$S\neq \emptyset$, some comparison keys are revealed and it causes increased
leakage. Our basic MC-ORE scheme is described as follows:

\begin{description}
\item [\tb{MC-ORE.Setup}($1^\lambda, N$).] This algorithm first generates
    bilinear groups $\G, \hat{\G}, \G_T$ of prime order $p$ with group
    generators $g \in \G$ and $\hat{g} \in \hat{\G}$. It chooses a random
    exponent $s_j \in \Z_p$ for all $j \in [N]$ and outputs a master key
    $MK = \{ s_j \}_{j \in [N]}$ and public parameters $PP = ((p, \G,
    \hat{\G}, \G_T, e), g, \hat{g}, H)$ where $H:\bits^* \to \G$ is a
    full-domain hash function.

\item [\tb{MC-ORE.GenKey}($j, MK, PP$).] Let $MK = \{ s_1, \cdots, s_N \}$.
    It outputs a secret key $SK_j = s_j$.

\item [\tb{MC-ORE.Encrypt}($m, SK_j, PP$).] Let $m = x_1 x_2 \cdots x_n \in
    \bits^n$ and $SK_j = s_j$. For each $i \in [n]$, it computes $C_{i,0} =
    H( \mb{prefix}(m, i) \| 0 x_i )^{s_j}$ and $C_{i,1} = H( \mb{prefix}(m,
    i) \| (0 x_i + 1) )^{s_j}$ where $\|$ is the concatenation of two bit
    strings. It outputs a ciphertext $CT_j = \big( \{ C_{i,0}, C_{i,1}
    \}_{i \in [n]} \big)$.

\item [\tb{MC-ORE.Compare}($CT_j, CT'_j, PP$).] For the same client index
    $j$, let $CT_j = ( \{ C_{i,0}, C_{i,1} \}_{i \in [n]} )$ and $CT'_j =
    ( \{C'_{i,0}, C'_{i,1}\}_{i \in [n]} )$.
    It first finds the smallest index $i^*$ such that $C_{i^*,0} \neq
    C'_{i^*,0}$ by sequentially comparing $C_{i,0}$ and $C'_{i,0}$. If such
    index $i^*$ exists and $C_{i^*,1} = C'_{i^*,0}$ holds, then it outputs
    1. If such index $i^*$ exists and $C_{i^*,0} = C'_{i^*,1}$, then it
    outputs 0. If no such index $i^*$ exists, then it outputs 0.

\item [\tb{MC-ORE.GenCmpKey}($j, k, MK, PP$).] Let $s_j$ and $s_k$ be the
    secret keys of client indices $j$ and $k$. It chooses a random exponent
    $r \in \Z_p$ and computes $K_0 = \hat{g}^{r s_j}, K_1 = \hat{g}^{r s_k}$.
    It outputs a comparison key $CK_{j,k} = ( K_0, K_1 )$.

\item [\tb{MC-ORE.CompareMC}($CT_j, CT'_k, CK_{j,k}, PP$).] Let $CT_j = (
    \{ C_{i,0}, C_{i,1} \}_{i \in [n]} )$ and $CT'_k = ( \{ C'_{i,0},
    C'_{i,1} \}_{i \in [n]} )$. Let $CK_{j,k} = (K_0, K_1)$.
    It first finds the smallest index $i^*$ such that $e(C_{i^*,0}, K_1)
    \neq e(C'_{i^*,0}, K_0)$ by sequentially comparing $e(C_{i,0}, K_1)$
    and $e(C'_{i,0}, K_0)$. If such index $i^*$ exists and $e(C_{i^*,1},
    K_1) = e(C'_{i^*,0}, K_0)$ holds, then it outputs 1. If such index
    $i^*$ exists and $e(C_{i^*,0}, K_1) = e(C'_{i^*,1}, K_0)$, then it
    outputs 0. If no such index $i^*$ exists, then it outputs 0.
\end{description}

\subsection{Correctness}

To show the correctness of the above scheme, we define encoding functions
$E_0, E_1$ that take $(i, m)$ as input and output the encoded $i$-th bit of
$m = x_1\cdots x_n\in\{0,1\}^n$ as follows:
    \begin{align*}
    E_0(i,m) = \mb{prefix}(m,i) \| 0 x_i,~ E_1(i,m) = \mb{prefix}(m,i) \| (0x_i+1).
    \end{align*}
The encoding functions satisfy the following conditions. If $m = m'$,
$E_0(i,m) = E_0(i,m')$ holds for all $i \in [n]$. If $m < m'$ and $i^*$ is
the smallest index such that $x_{i^*} \neq x'_{i^*}$, then $E_0(i,m) =
E_0(i,m')$ holds for all $i < i^*$ and $E_0(i,m) \neq E_0(i,m')$ holds for
all $i \ge i^*$, and especially, $E_1(i^*,m) = E_0(i^*,m')$ holds.

\vs\noindent Let $SK_j = s_j$ be the secret key of a client index $j$ and
$CT_j = ( \{ C_{i,0}, C_{i,1} \}_{i \in [n]} )$ and $CT'_j = ( \{ C'_{i,0},
C'_{i,1} \}_{i \in [n]} )$ be ciphertexts on messages $m = x_1 x_2 \cdots x_n
\in \bits^n$ and $m' = x'_1 x'_2 \cdots x'_n \in \bits^n$. If $m < m'$, there
must be the smallest index $i^*$ such that $x_i = x'_i$ for all $i < i^*$ and
$x_{i^*} \neq x'_{i^*}$. Thus, we have that
\begin{align*}
    C_{i,0}
        &= H(E_0(i,m))^{s_j} = H(E_0(i,m'))^{s_j} = C'_{i,0} ~~\forall i < i^*
         ~~\text{and} \\
    C_{i^*,1}
        &= H(E_1(i^*,m))^{s_j} = H(E_0(i^*,m'))^{s_j} = C'_{i^*,0}.
\end{align*}

\vs\noindent Let $SK_j = s_j$ and $SK_k = s_k$ be the secret keys of two
client indices $j$ and $k$, and $CK_{j,k} = (K_0, K_1) = (\hat{g}^{r s_j},
\hat{g}^{r s_k})$ be the comparison key. Let $CT_j = ( \{ C_{i,0}, C_{i,1}
\}_{i \in [n]} )$ and $CT'_k = ( \{ C'_{i,0}, C'_{i,1} \}_{i \in [n]} )$ be
ciphertexts on messages $m$ and $m'$. If $m < m'$, there must be the smallest
index $i^*$ such that $x_i = x'_i$ for all $i < i^*$ and $x_{i^*} \neq
x'_{i^*}$. Thus, we have that
\begin{align*}
    e(C_{i,0}, K_1)
        &= e(H(E_0(i,m))^{s_j}, \hat{g}^{r s_k})
         = e(H(E_0(i,m)), \hat{g})^{r s_j s_k} \\
        &= e(H(E_0(i,m'))^{s_k}, \hat{g}^{r s_j})
         = e(C'_{i,0}, K_0) ~~\forall i < i^*
         ~~\text{and} \\
    e(C_{i^*,1}, K_1)
        &= e(H(E_1(i^*,m))^{s_j}, \hat{g}^{r s_k})
         = e(H(E_1(i^*,m)), \hat{g})^{r s_j s_k} \\
        &= e(H(E_0(i^*,m'))^{s_k}, \hat{g}^{r s_j})
         = e(C'_{i^*,0}, K_0).
\end{align*}

\subsection{Security Analysis}

We prove the security of the basic MC-ORE scheme with the leakage function
$\mc{L}_S$ in the ST-SIM security model. We define a sequence of experiments
from $\mb{H}_0$ corresponding to the real experiment to $\mb{H}_3$
corresponding to the ideal experiment and show that the outputs of two
experiments are indistinguishable. At first, the ciphertexts of clients whose
comparison keys are not exposed are randomly generated. In the next
experiment, the ciphertexts of clients whose comparison keys are exposed are
generated with random values. Finally, in the last experiment $\mb{H}_3$, the
ciphertexts are simulated with respect to the leakage function $\mc{L}_S$,
and consequently $\mb{H}_3$ corresponds to the ideal experiment. The details
are given as follows.

\begin{theorem} \label{thm:mc-ore-1}
The basic MC-ORE scheme is ST-SIM secure with the leakage function $\mc{L}_S$
in the random oracle model if the XDH assumption holds.
\end{theorem}

\begin{proof}
We prove the security of the basic MC-ORE scheme through a sequence of hybrid
experiments. The first experiment is defined as the real MC-ORE security
experiment and the last one is defined as the ideal experiment with the
leakage function $\mc{L}_S$ in which the adversary has no advantage. The
hybrid experiments $\mb{H}_0, \mb{H}_1, \mb{H}_2$, and $\mb{H}_3$ are defined
as follows:
\begin{description}
\item [$\mb{H}_0$]: This experiment corresponds to the real world
    experiment.

\item [$\mb{H}_1$]: This experiment is similar to $\mb{H}_0$ except that
    the ciphertext $CT_j$ such that $(j, j')\notin S$ for any client index
    $j'$ is generated by using random elements.

\item [$\mb{H}_2$]: This experiment is similar to $\mb{H}_1$ except that
    the ciphertext $CT_j$ such that $(j, j')\in S$ for some client index
    $j'$ is generated by using random elements.

\item [$\mb{H}_3$]: In this experiment, the ciphertexts are generated with
    the leakage function $\mc{L}_S$ and the rest are same to $\mb{H}_2$.
    This experiment corresponds to the ideal world experiment.
\end{description}

From the following Lemmas \ref{lem:mc-ore-1-h0-h1}, \ref{lem:mc-ore-1-h1-h2},
and \ref{lem:mc-ore-1-h2-h3} that claim the indistinguishability of the above
experiments, we have that $\mb{H}_0$ and $\mb{H}_3$ are computationally
indistinguishable.
\end{proof}

Before we present additional Lemmas for the proof of the above theorem, we
define the encoded messages $E_0(k,m) = \mb{prefix}(m,k) \| 0x_k$ and
$E_1(k,m) = \mb{prefix}(m,k) \| (0x_k + 1)$ where $m = x_1 \cdots x_n \in
\bits^n$. In addition, we introduce the multi-external Diffie-Hellman
assumption.

\begin{assumption}[Multi-External Diffie-Hellman, mXDH]
Let $(p, \G, \hat{\G}, \G_T, e)$ be a tuple randomly generated by
$\mc{G}_{as}(1^\lambda)$ where $p$ is a prime order of the groups. Let $g,
\hat{g}$ be random generators of groups $\G, \hat{\G}$, respectively. The
mXDH assumption is that if the challenge tuple
\begin{align*}
    D = \big( (p, \G, \hat{\G}, \G_T, e), g, \hat{g}, g^a,
        \{ g^{b_{i,1}}, \cdots, g^{b_{i,n}} \}_{i \in [t]} \big) \mbox{ and } T
\end{align*}
are given, no PPT algorithm $\mc{A}$ can distinguish $T = T_0 = \big( \{
g^{ab_{i,1}}, \cdots, g^{ab_{i,n}} \}_{i\in [t]} \big)$ from $T = T_1 = \big
\{(g^{c_{i,1}}, \cdots, g^{c_{i,n}} \}_{i\in [t]} \big)$ with more than a
negligible advantage. The advantage of $\mc{A}$ is defined as
    $\Adv_{\mc{A}}^{mXDH} (\lambda) = \big|
    \Pr[\mc{A}(D,T_0) = 0] - \Pr[\mc{A}(D,T_1) = 0] \big|$
where the probability is taken over random choices of $a, (b_{i,1}, \cdots,
b_{i,n}), (c_{i,1}, \cdots, c_{i,n}) \in \Z_p$ for all $i \in [t]$.
\end{assumption}

This mXDH assumption is equivalent to the XDH assumption since the challenge
tuple of mXDH assumption can be obtained from the XDH assumption by using the
random self-reducibility property \cite{NaorR04}.

\begin{lemma} \label{lem:mc-ore-1-h0-h1}
The hybrid experiments $\mb{H}_0$ and $\mb{H}_1$ are computationally
indistinguishable to the polynomial-time adversary assuming that the mXDH
assumption holds.
\end{lemma}

\begin{proof}
To prove this lemma, we additionally define a sequence of hybrid experiments
$\mb{H}_0 = \mb{H}_{0,0}, \mb{H}_{0,1}, \ldots, \lb \mb{H}_{0,\tilde{q}} =
\mb{H}_1$ for some $\tilde{q}$ as follows.
\begin{description}
\item [$\mb{H}_{0, \mu}$]: Let $I = ( j_1, \cdots, j_q )$ be a tuple of
    challenge client index. For all $j_i \in I$ such that $(j_i, *) \notin
    S$, let $j^*_1,\cdots, j^*_{\tilde{q}} \in I$ be distinct client
    indices where $\tilde{q} \le q$. Let $\mb{SI}_\mu = \{ i \in [q] : j_i
    = j^*_\mu \}$ be an index set of same client indices where $\mu \in
    [\tilde{q}]$. In this experiment, we change the generation of the
    $\mu$-th ciphertext set with the index set $\mb{SI}_\mu$. If $\ell \le
    \mu$, the ciphertexts in the $\ell$-th ciphertext set with
    $\mb{SI}_\ell$ are changed to be random elements. Otherwise, the
    ciphertexts in the $\ell$-th ciphertext set with $\mb{SI}_\ell$ are
    generated by running the normal encryption algorithm. Note that the
    ciphertexts with the client index $j_i$ such that $(j_i, *)\in S$ in
    $\mb{H}_{0,\mu-1}$ and $\mb{H}_{0,\mu}$ are equally generated by
    running the normal encryption algorithm.
\end{description}

Without loss of generality, we assume that $(j^*_\mu, *) \notin S$. Suppose
there exists an adversary $\mc{A}$ that distinguishes $\mb{H}_{0, \mu-1}$
from $\mb{H}_{0, \mu}$ with non-negligible advantage. A simulator $\mc{B}$
that solves the mXDH assumption using $\mc{A}$ is given: a challenge tuple $D
= \big( (p, \G, \hat{\G}, \G_T, e), g, \hat{g}, g^a, \{ g^{b_{i,1}}, \cdots,
g^{b_{i,2n}} \}_{i \in [t]} \big)$ and $T = \big( \{ X_{i,1}, \cdots,
X_{i,2n} \}_{i \in [t]} \big)$. $\mc{B}$ interacts with $\mc{A}$ as follows.

Let $( st_{\mc{A}}, S, ((j_1, m_1), \cdots, (j_q, m_q)) )$ be the output of
$\mc{A}$ and $\mb{SI}_\mu$ be the target index set of $j^*_\mu$. The
simulator $\mc{B}$ first sets the secret keys of all clients except the
target client. For each $j \neq j^*_\mu$, it chooses a random exponent $s_j$
and sets $SK_j = s_j$. For the target client index $j^*_\mu$, it implicitly
sets $SK_{j^*_\mu} = a$. Now, $\mc{B}$ can generate any comparison key
$CK_{j,k}$ for all tuple $(j,k) \in S$ since it knows secret keys $s_j$ and
$s_k$ if $(j,k) \in S$.

To handle hash queries, $\mc{B}$ maintains a random oracle table $T_H$ for
the consistency of a simulation.
Initially, $\mc{B}$ fixes some hash queries for the simulation of the
ciphertext with the challenge tuple $(j_i, m_i)$ such that $i \in
\mb{SI}_\mu$, which is output of $\mc{A}$. For the first message $m_1$,
$\mc{B}$ sets $h_{k,0} = g^{b_{1, 2k-1}}, h_{k,1} = g^{b_{1,2k}}$ and adds
the tuples $(E_0(k, m_1), h_{k,0})$ and $(E_1(k, m_1), h_{k,1})$ to the table
$T_H$ for all $k \in [n]$.
For each message $m_i$, $\mc{B}$ first finds the biggest index $d =
\mb{ind}(m_i, m_{i'})$ for any $i' < i$ and finds tuples $(E_0(k, m_{i'}),
h'_{k,0}), (E_1(k, m_{i'}), h'_{k,1})$ from $T_H$ for all $k\in [d]$. It sets
$h_{k,0} = h'_{k,0}, h_{k,1} = h'_{k,1}$ for all $k \in [d-1]$ since $E_0(k,
m_{i'}) = E_0(k, m_i)$ and $E_1(k, m_{i'}) = E_1(k, m_i)$. If $\mb{cmp}(m_i,
m_{i'}) = 1$, then $\mc{B}$ sets $h_{d,0} = g^{b_{i, 2d-1}}, h_{d,1} =
h'_{d,0}$ and otherwise, it sets $h_{d,0} = h'_{d,1}, h_{d,1} = g^{b_{i,
2d}}$. Next, it sets $h_{k,0} = g^{b_{i, 2k-1}}, h_{k,1} = g^{b_{i, 2k}}$ for
all $k \in [d+1, n]$. It adds the tuples $(E_0(k, m_i), h_{k,0})$ and
$(E_1(k, m_i), h_{k,1})$ to the table $T_H$ for all $k \in [n]$.
After that, if a random oracle query for an encoded message $E_\beta(k,m)$ is
requested for each $\beta \in \bits$, $\mc{B}$ first finds a tuple
$(E_\beta(k, m), h)$ on the table $T_H$. If the tuple does not exist, then it
chooses a random element $h \in \G$ and adds the tuple $(E_\beta(k,m), h)$ to
$T_H$. Finally it gives $h$ to $\mc{A}$ as a response.

To handle the creation of ciphertexts, $\mc{B}$ carefully uses the hash table
and the challenge elements in the assumption. Let $((j_1, m_1), \cdots, (j_q,
m_q))$ be the challenge tuples. If $(j_i, *) \in S$, then $\mc{B}$ simply
creates a ciphertext by running the \tb{MC-ORE.Encrypt} algorithm with hash
queries since it knows the secret key $s_{j_i}$. If $(j_i, *) \notin S$, then
it means that $i \in \mb{SI}_{\ell}$ for some $\ell \in [\tilde{q}]$.
$\mc{B}$ creates a set of ciphertexts with the index set $\mb{SI}_\ell$ for
each $\ell \in [\tilde{q}]$ as follows:
\begin{itemize}
\item \tb{Case} $\ell < \mu$: $\mc{B}$ creates the ciphertext for the index
    $i \in \mb{SI}_\ell$ sequentially.
    For the smallest index $i \in \mb{SI}_\ell$, it chooses random elements
    $R_{k,0}, R_{k,1} \in \G$ for all $k \in [n]$ and creates $CT_{j_i} = (
    \{ R_{k,0}, R_{k,1} \}_{k \in [n]} )$.
    For the next index $i$, it first finds the biggest index $d =
    \mb{ind}(m_i, m_{i'})$ for any $i' < i$. It sets $C_{k,0} = C'_{k,0},
    C_{k,1} = C'_{k,1}$ for all $k \in [d-1]$ where $CT_{j_{i'}} = ( \{
    C'_{k,0}, C'_{k,1} \}_{k \in [n]} )$. If $\mb{cmp}(m_i, m_{i'}) = 1$,
    then it chooses a random element $R_{d,0}\in \G$ and sets $C_{d,0} =
    R_{d,0}, C_{d,1} = C'_{d,0}$. Otherwise, it chooses a random element
    $R_{d,1}\in \G$ and sets $C_{d,0} = C'_{d,1}, C_{d,1} = R_{d,1}$. Next,
    it chooses random elements $R_{k,0}, R_{k,1}\in \G$ and sets $C_{k,0} =
    R_{k,0}, C_{k,1} = R_{k,1}$ for all $k \in [d+1, n]$. It creates the
    ciphertext $CT_{j_i} = ( \{ C_{k,0}, C_{k,1} \}_{k \in [n]} )$.
    At last, it creates the $\ell$-th ciphertext set $\mb{CT}_{SI_\ell} = (
    \{ CT_{j_i} \}_{i\in SI_\ell} )$.

\item \tb{Case} $\ell = \mu$: $\mc{B}$ creates the ciphertext for the index
    $i \in \mb{SI}_\mu$ sequentially.
    For the smallest index $i \in \mb{SI}_\ell$, $\mc{B}$ sets $C_{k,0} =
    X_{1, 2k-1}, C_{k,1} = X_{1, 2k}$ for all $k \in [n]$ and creates the
    ciphertext $CT_{j_i} = ( \{ C_{k,0}, C_{k,1} \}_{k \in [n]} )$.
    For the next index $i$, $\mc{B}$ first finds the biggest index $d =
    \mb{ind}(m_i, m_{i'})$ for any $i' < i$. It sets $C_{k,0} = C'_{k,0},
    C_{k,1} = C'_{k,1}$ for all $k \in [d-1]$ where $CT_{j_{i'}} = ( \{
    C'_{k,0}, C'_{k,1} \}_{k \in [n]} )$. If $\mb{cmp}(m_i, m_{i'}) = 1$,
    then it sets $C_{d,0} = X_{i, 2d-1}, C_{d,1} = C'_{d,0}$ and otherwise,
    it sets $C_{d,0} = C'_{d,1}, C_{d,1} = X_{i, 2d}$. Next, it sets
    $C_{k,0} = X_{i, 2k-1}, C_{k,1} = X_{i, 2k}$ for all $k \in [d+1, n]$.
    It creates the ciphertext $CT_{j_i} = ( \{ C_{k,0}, C_{k,1} \}_{k \in
    [n]} )$.
    At last, it creates the $\mu$-th ciphertext set $\mb{CT}_{SI_\mu} = (
    \{ CT_{j_i} \}_{i\in SI_\mu} )$. Note that it does not know the secret
    key $a$.

\item \tb{Case} $\ell > \mu$: It creates the ciphertext set
    $\mb{CT}_{SI_\ell}$ by running the \tb{MC-ORE.Encrypt} algorithm with
    hash queries.
\end{itemize}
If $T = \big( (g^{ab_{1,1}}, \cdots, g^{ab_{1,2n}}), \ldots, (g^{ab_{t,1}},
\cdots, g^{ab_{t, 2n}}) \big)$, then $\mb{CT}_{SI_\mu}$ are ciphertexts in
$\mb{H}_{0,\mu-1}$. Otherwise, $\mb{CT}_{SI_\mu}$ are ciphertexts in
$\mb{H}_{0,\mu}$. By the mXDH assumption, two experiments $\mb{H}_{0,\mu-1}$
and $\mb{H}_{0,\mu}$ are computationally indistinguishable.
\end{proof}

\begin{lemma} \label{lem:mc-ore-1-h1-h2}
The hybrid experiments $\mb{H}_1$ and $\mb{H}_2$ are computationally
indistinguishable to the polynomial-time adversary assuming that the mXDH
assumption holds.
\end{lemma}

\begin{proof}
We additionally define a sequence of hybrid experiments $\mb{H}_1 =
\mb{H}_{1,0}, \mb{H}_{1,1}, \cdots, \mb{H}_{1,\tilde{q}} = \mb{H}_2$ for some
$\tilde{q}$ as follows.
\begin{description}
\item [$\mb{H}_{1, \mu}$]: Let $I = ( j_1, \cdots, j_q )$ be a tuple of
    challenge client index and let $j, j' \in I$ be co-related indices if
    $(j,j')\in S$ or there exist $\{ k_i \}_{i\in[n]} \subseteq I$ such
    that $(j, k_1), (k_1, k_2), \cdots, (k_{n-1}, k_n), (k_n, j') \in S$
    for any $n \in [ q-2]$. Let $\mb{RI}_\mu = \{ i \in [q] : j_i\text{s
    are co-related indices} \}$ be an index set of co-related client
    indices where $\mu \in [\tilde{q}]$. In this experiment, we change the
    generation of the $\mu$-th ciphertext set with the index set
    $\mb{RI}_\mu$. If $\ell \le \mu$, the ciphertexts in the $\ell$-th
    ciphertext set with $\mb{RI}_\ell$ are changed to be random elements.
    Otherwise, the ciphertexts in the $\ell$-th ciphertext set with
    $\mb{RI}_\ell$ are generated by running the normal encryption
    algorithm. Note that the ciphertexts with the client index $j_i$ such
    that $(j_i, *)\notin S$ in $\mb{H}_{1,\mu-1}$ and $\mb{H}_{1,\mu}$ are
    equally generated by using random elements.
\end{description}

Suppose there exists an adversary $\mc{A}$ that distinguishes $\mb{H}_{1,
\mu-1}$ from $\mb{H}_{1, \mu}$ with non-negligible advantage. A simulator
$\mc{B}$ that solves the mXDH assumption using $\mc{A}$ is given: a challenge
tuple $D = \big( (p, \G, \hat{\G}, \G_T, e), g,\lb \hat{g}, g^a, \{
g^{b_{i,1}}, \cdots, g^{b_{i,2n}} \}_{i \in [t]} \big)$ and $T = \big( \{
X_{i,1}, \cdots, X_{i,2n} \}_{i \in [t]} \big)$. $\mc{B}$ interacts with
$\mc{A}$ as follows.

Let $( st_{\mc{A}}, S, ((j_1, m_1), \cdots, (j_q, m_q)) )$ be the output of
$\mc{A}$ and $\mb{RI}_\mu$ be the target index set. The simulator $\mc{B}$
first sets the secret keys of clients as follows. For each $j = j_i$, if $i
\notin \mb{RI}_\mu$, it chooses a random exponent $s_j \in \Z_p$ and sets
$SK_j = s_j$. Otherwise, it chooses a random exponent $s_j \in \Z_p$ and
implicitly sets $SK_j = as_j$. Then, $\mc{B}$ can generate a comparison key
$CK_{j,k} = ( \hat{g}^{r s_j}, \hat{g}^{r s_k} )$ for each tuple $(j,k) \in
S$ with the help of a random exponent $r \in \Z_p$, though it does not know
$a$.

To handle hash queries, $\mc{B}$ maintains a random oracle table $T_H$ for
the consistency of a simulation. This simulation is same to the proof of the
Theorem \ref{lem:mc-ore-1-h0-h1} except that $\mc{B}$ fixes some hash queries
for the simulation of the ciphertext with the challenge tuple $(j_i, m_i)$
such that $i \in \mb{RI}_\mu$.

To handle the creation of ciphertexts, $\mc{B}$ carefully uses the hash table
and the challenge elements in the assumption. Let $((j_1, m_1), \cdots, (j_q,
m_q))$ be the challenge tuples. If $(j_i, *) \notin S$, then $\mc{B}$ creates
a ciphertext by using random elements as in $\mb{H}_{1,\mu-1}$. If $(j_i, *)
\in S$, then it means that $i \in \mb{RI}_\ell$ for some $\ell \in
[\tilde{q}]$. $\mc{B}$ creates a set of ciphertexts with the index set
$\mb{RI}_\ell$ for each $\ell \in [\tilde{q}]$ as follows:

\begin{itemize}
\item \tb{Case} $\ell < \mu$: $\mc{B}$ creates the ciphertext for the index
    $i \in \mb{RI}_\ell$ sequentially.
    For the smallest index $i \in \mb{RI}_\ell$, it chooses random elements
    $R_{k,0}, R_{k,1} \in \G$ and computes $C_{k,0} = R_{k,0}^{s_{j_i}},
    C_{k,1} = R_{k,1}^{s_{j_i}}$ for all $k \in [n]$. It creates $CT_{j_i}
    = ( \{ C_{k,0}, C_{k,1} \}_{k \in [n]} )$.
    For the next index $i$, $\mc{B}$ first finds the biggest index $d =
    \mb{ind}(m_i, m_{i'})$ for any $i' < i$ and computes $s = s_{j_i} /
    s_{j_{i'}}$. It computes $C_{k,0} = {C'}_{k,0}^s, C_{k,1} =
    {C'}_{k,1}^s$ for all $k \in [d-1]$ where $CT_{j_{i'}} = ( \{ C'_{k,0},
    C'_{k,1} \}_{k \in [n]} )$. If $\mb{cmp}(m_i, m_{i'}) = 1$, then it
    chooses a random element $R_{d,0}\in \G$ and computes $C_{d,0} =
    R_{d,0}^{s_{j_i}}, C_{d,1} = {C'}_{d,0}^s$. Otherwise, it chooses a
    random element $R_{d,1}\in \G$ and computes $C_{d,0} = {C'}_{d,1}^s,
    C_{d,1} = R_{d,1}^{s_{j_i}}$. Next, it chooses random elements
    $R_{k,0}, R_{k,1}\in \G$ and computes $C_{k,0} = R_{k,0}^{s_{j_i}},
    C_{k,1} = R_{k,1}^{s_{j_i}}$ for all $k \in [d+1, n]$. It creates the
    ciphertext $CT_{j_i} = ( \{ C_{k,0}, C_{k,1} \}_{k \in [n]} )$.
    At last, it creates the $\ell$-th ciphertext set $\mb{CT}_{RI_\ell} =
    \{ CT_{j_i} \}_{i\in RI_\ell}$.

\item \tb{Case} $\ell = \mu$: $\mc{B}$ creates the ciphertext for the index
    $i \in \mb{RI}_\mu$ sequentially.
    For the smallest index $i \in \mb{RI}_\mu$, $\mc{B}$ computes $C_{k,0}
    = X_{1,2k-1}^{s_{j_i}}, C_{k,1} = X_{1,2k}^{s_{j_i}}$ for all $k \in
    [n]$ and creates the ciphertext $CT_{j_i} = ( \{ C_{k,0}, C_{k,1} \}_{k
    \in [n]} )$.
    For the next index $i$, $\mc{B}$ first finds the biggest index $d =
    \mb{ind}(m_i, m_{i'})$ for any $i' < i$ and it computes $s = s_{j_i} /
    s_{j_{i'}}$. It computes $C_{k,0} = {C'}_{k,0}^s, C_{k,1} =
    {C'}_{k,1}^s$ for all $k \in [d-1]$ where $CT_{j_{i'}} = ( \{ C'_{k,0},
    C'_{k,1} \}_{k \in [n]} )$. If $\mb{cmp}(m_t, m_{t'}) = 1$, then
    $\mc{B}$ computes $C_{d,0} = X_{i, 2d-1}^{s_{j_i}}, C_{d,1} =
    {C'}_{d,0}^s$ and otherwise, it computes $C_{d,0} = {C'}_{d,1}^s,
    C_{d,1} = X_{i,2d}^{s_{j_i}}$. Next, it computes $C_{k,0} = X_{i,
    2k-1}^{s_{j_i}}, C_{k,1} = X_{i, 2k}^{s_{j_i}}$ for all $k \in [d+1,
    n]$. Then, it creates the ciphertext $CT_{j_i} = ( \{ C_{k,0}, C_{k,1}
    \}_{k \in [n]} )$.
    At last, it creates the $\mu$-th ciphertext set $\mb{CT}_{RI_\mu} = \{
    CT_{j_i} \}_{i \in RI_\mu}$. Note that it does not know the secret key
    $a$.

\item \tb{Case} $\ell > \mu$: It creates the ciphertext set
    $\mb{CT}_{RI_\ell}$ by running the \tb{MC-ORE.Encrypt} algorithm with
    hash queries.

\end{itemize}
If $T = \big( (g^{ab_{1,1}}, \cdots, g^{ab_{1,2n}}), \ldots, (g^{ab_{t,1}},
\cdots, g^{ab_{t, 2n}}) \big)$, then $\mb{CT}_{RI_\mu}$ are ciphertexts in
$\mb{H}_{1,\mu-1}$. Otherwise, $\mb{CT}_{RI_\mu}$ are ciphertexts in
$\mb{H}_{1,\mu}$. By the mXDH assumption, two experiments $\mb{H}_{1,\mu-1}$
and $\mb{H}_{1,\mu}$ are computationally indistinguishable.
\end{proof}

\begin{lemma} \label{lem:mc-ore-1-h2-h3}
The hybrid experiments $\mb{H}_2$ and $\mb{H}_3$ are indistinguishable to the
polynomial-time adversary with the leakage function $\mc{L}_S$ in the random
oracle model.
\end{lemma}

\begin{proof}
Suppose there exists an adversary $\mc{A}$ that distinguishes $\mb{H}_2$ from
$\mb{H}_3$ with non-negligible advantage. We construct an efficient simulator
$\mc{B}$ for which the two distributions $\mb{H}_2$ and $\mb{H}_3$ are
statistically indistinguishable.

Let $( st_{\mc{A}}, S, ((j_1, m_1), \cdots, (j_q, m_q)) )$ be the output of
$\mc{A}$. $\mc{B}$ first outputs random public parameters $PP$ with the
initial state $st_{\mc{B}}$. It selects a random secret key $SK_j = s_j \in
\Z_p$ for each client index $j \in [N]$ and it can generate any comparison
key $CK_{j,k}$ for $(j,k) \in S$ since it knows all secret keys.

To handle hash queries, $\mc{B}$ maintains a random oracle table $T_{H}$ for
consistency of the simulation. If a random oracle query for $E_\beta(k, m)$
is requested for each $\beta \in \bits$, $\mc{B}$ first finds the tuple
$(E_\beta(k, m), h)$ from the table $T_H$. If the tuple does not exist, then
it chooses a random element $h \in \G$ and adds the tuple $(E_\beta(k,m), h)$
to $T_H$. Finally it gives $h$ to $\mc{A}$ as a response.

To handle the creation of ciphertexts, $\mc{B}$ also maintains a ciphertext
table $T_{CT}$ for consistency of the simulation. Let $I = ( j_1, \cdots, j_q
)$ be a tuple of challenge client index. For all $j_i \in I$ such that $(j_i,
*) \notin S$, let $j^*_1,\cdots, j^*_{\tilde{q}_1} \in I$ be distinct client
indices and $\mb{SI}_\mu = \{ i \in [q] : j_i = j^*_\mu \}$ be an index set
of same client indices where $\mu \in [\tilde{q}_1]$. For all $j_i \in I$
such that $(j_i, *)\in S$, let $\mb{RI}_\mu = \{ i \in [q] : j_i\text{s are
co-related indices} \}$ be an index set of co-related client indices where
$\mu \in [\tilde{q}_2]$. $\mc{B}$ simulates the creation of a set of
ciphertexts with a client index set $\mb{SI}_\mu$ or $\mb{RI}_\mu$ by using
$st_{\mc{B}}$ and $\mc{L}_S((j_1,m_1), \cdots ,(j_q,m_q))$ as follows:

\begin{itemize}
\item For the creation of the ciphertexts with each set $\mb{SI}_\mu$,
    $\mc{B}$ initiates the ciphertext table $T_{CT}$. For the smallest
    index $i \in \mb{SI}_\mu$, $\mc{B}$ chooses random elements $(c_{k,0},
    c_{k,1}) \in \G \times \G$ and sets $(C_{k,0}, C_{k,1}) = ( c_{k,0},
    c_{k,1} )$ for all $k \in [n]$. It adds the tuple $(i, (c_{1,0},
    c_{1,1}), \ldots, (c_{n,0}, c_{n,1}) )$ to $T_{CT}$ and creates
    $CT_{j_i} = ( \{ C_{k,0}, C_{k,1} \}_{k \in [n]} )$.
    For the next index $i\in \mb{SI}_\ell$, it creates the ciphertext
    sequentially as follows. It first finds the biggest index $b =
    \mb{ind}(m_i, m_{i'})$ for any $i' < i$ and then finds a tuple $(i',
    (c'_{1,0}, c'_{1,1}), \ldots, (c'_{n,0}, c'_{n,1}) )$ from the table
    $T_{CT}$. If $b = n+1$, it sets $(c_{k,0}, c_{k,1}) = (c'_{k,0},
    c'_{k,1})$ for all $k \in [n]$. If not, it proceeds the following steps:
    \begin{enumerate}
    \item It sets $(c_{k,0}, c_{k,1}) = (c'_{k,0}, c'_{k,1})$ for all $k
        \in [b-1]$.
    \item It chooses random elements $(c_{k,0}, c_{k,1}) \in \G \times
        \G$ for all $k \in [b+1, n]$.
    \item If $\mb{cmp}(m_t, m_{t'}) = 1$, it sets $c_{b,1} = c'_{b,0}$
        and chooses a random element $c_{b,0} \in \G$. Otherwise, it sets
        $c_{b,0} = c'_{b,1}$ and chooses a random element $c_{b,1} \in
        \G$.
    \end{enumerate}
    Then, $\mc{B}$ creates $CT_{j_i} = ( \{ c_{k,0}, c_{k,1} \}_{k \in [n]}
    )$ and adds the tuple $(i, (c_{1,0}, c_{1,1}), \ldots, (c_{n,0},
    c_{n,1}) )$ to $T_{CT}$. At last, it creates the ciphertext set
    $\mb{CT}_{SI_\mu} = (\{ CT_{j_i} \}_{i\in SI_\mu})$.

\item For the creation of the ciphertext with each set $\mb{RI}_\mu$,
    $\mc{B}$ initiates the ciphertext table $T_{CT}$. For the smallest
    index $i \in \mb{RI}_\mu$, $\mc{B}$ chooses random elements $(c_{k,0},
    c_{k,1}) \in \G \times \G$ and computes $(C_{k,0}, C_{k,1}) = (
    c_{k,0}^{s_{j_i}}, c_{k,1}^{s_{j_i}} )$ for all $k \in [n]$. It adds
    the tuple $(i, (c_{1,0}, c_{1,1}), \ldots, (c_{n,0}, c_{n,1}) )$ to
    $T_{CT}$ and creates $CT_{j_i} = ( \{ C_{k,0}, C_{k,1} \}_{k \in [n]}
    )$.
    For the next index $i\in \mb{RI}_\mu$, it creates the ciphertext
    sequentially as follows. It first finds the biggest index $b =
    \mb{ind}(m_i, m_{i'})$ for any $i' < i$ and then finds a tuple $(i',
    (c'_{1,0}, c'_{1,1}), \ldots, (c'_{n,0}, c'_{n,1}) )$ from the table
    $T_{CT}$. If $b = n+1$, it sets $(c_{k,0}, c_{k,1}) = (c'_{k,0},
    c'_{k,1})$ for all $k \in [n]$. If not, it proceeds the steps $1) - 3)$
    described in the creation of the $\mb{CT}_{SI_\mu}$. Then, $\mc{B}$
    computes $(C_{k,0}, C_{k,1}) = ( c_{k,0}^{s_{j_i}}, c_{k,1}^{s_{j_i}})$
    for all $k\in [n]$. It creates $CT_{j_i} = ( \{ C_{k,0}, C_{k,1} \}_{k
    \in [n]} )$ and adds the tuple $(i, (c_{1,0}, c_{1,1}), \ldots,
    (c_{n,0}, c_{n,1}) )$ to $T_{CT}$. At last, it creates the ciphertext
    set $\mb{CT}_{RI_\mu} = (\{ CT_{j_i} \}_{i\in RI_\mu})$.
\end{itemize}

\vs \noindent \tb{Correctness of the Simulation.} To show the correctness of
the simulation, we prove that the distributions $\big( (\mb{CT}_{SI_1},
\ldots, \mb{CT}_{SI_{\tilde{q}_1}}), (\mb{CT}_{RI_1}, \ldots,
\mb{CT}_{RI_{\tilde{q}_2}}) \big)$ and $\big( (\overline{\mb{CT}}_{SI_1},
\ldots, \overline{\mb{CT}}_{SI_{\tilde{q}_1}}), (\overline{\mb{CT}}_{RI_1},
\ldots, \overline{\mb{CT}}_{RI_{\tilde{q}_2}}) \big)$ of the ciphertexts
output in $H_2$ and $H_3$ are statistically indistinguishable and the outputs
of random oracle are properly simulated. We have to show that the following
conditions hold.
\begin{itemize}
\item $\forall \ell \in [\tilde{q}_1], \forall \ell' \in [\tilde{q}_2]$,
    $\mb{CT}_{SI_\ell}$ and $\mb{CT}_{RI_{\ell'}}$ are distributed
    independently.

\item $\forall \ell \in [\tilde{q}_1], \forall \ell' \in [\tilde{q}_2]$,
    $\mb{CT}_{SI_\ell} \equiv \overline{\mb{CT}}_{SI_\ell}$ and
    $\mb{CT}_{RI_{\ell'}} \equiv \overline{\mb{CT}}_{RI_{\ell'}}$.
\end{itemize}
The first condition is simply proved since each ciphertext for $SI_\ell$ and
$RI_{\ell'}$ are simulated independently. Next, we use induction to prove
that the second condition holds as follows.

\begin{itemize}
\item For each $\ell \in [\tilde{q}_1]$, let $\mb{CT}_{SI_\ell} = ( CT_1,
    \cdots, CT_t )$ and $\overline{\mb{CT}}_{SI_\ell} = ( \overline{CT}_1,
    \cdots, \overline{CT}_t )$. Obviously, the statement is true for $i =
    1$. Assume that it is true for $i-1$ and we must prove that $(CT_1,
    \ldots, CT_i) \equiv (\overline{CT}_1, \ldots, \overline{CT}_i)$.

    Suppose that $CT_i, CT_{i'}$ are the ciphertexts of $m, m'$ where $i' <
    i$. For the biggest index $b = \mb{ind}(m, m')$, if $b = n+1$, then
    $CT_i$ and $CT_{i'}$ are the ciphertexts of the same message. In the
    simulation, $\mc{B}$ finds the tuple $(-, (c'_{1,0}, c'_{1,1}), \ldots,
    (c'_{n,0}, c'_{n,1}))$ from the table $T_{CT}$ and uses it to simulate
    the ciphertext $CT_i$ by setting $(C_{k,0}, C_{k,1}) = ( {c'_{k,0}},
    {c'_{k,1}} )$ for all $k \in [n]$. Then, we have
        $$C_{k,0} = c_{k,0} = c'_{k,0} = C'_{k,0}~~\forall k\in [n].$$
    Otherwise, $m$ and $m'$ may have the same prefix of the length $b-1$.
    For $k \in [b-1]$, $(C_{k,0}, C_{k,1})$ has been simulated as the
    previous case and for $k \in [b+1, n]$, $(C_{k,0}, C_{k,1})$ has been
    simulated by using random elements. For the remain part $(C_{b,0},
    C_{b,1})$, $\mc{B}$ simulates $c_{b,1} = c'_{b,0}$ if $\mb{cmp}(m, m')
    = 1$. Then we have
            $$ C'_{b,0} = c'_{b,0} = c_{b,1} = C_{b,1}.$$
    Since we assumed that $CT_{i'}$ and $\overline{CT}_{i'}$ are
    identically distributed, by induction, $CT_i$ and $\overline{CT}_i$ are
    identically distributed.

\item For each $\ell \in [\tilde{q}_2]$, let $\mb{CT}_{RI_\ell} = ( CT_1,
    \cdots, CT_t )$ and $\overline{\mb{CT}}_{RI_\ell} = ( \overline{CT}_1,
    \cdots, \overline{CT}_t )$. Obviously, the statement is true for $i =
    1$. Assume that it is true for $i-1$ and we must prove that $(CT_1,
    \ldots, CT_i) \equiv (\overline{CT}_1, \ldots, \overline{CT}_i)$.

    Suppose that $CT_i, CT_{i'}$ are the ciphertexts of $(j, m), (j', m')$
    where $i' < i$. For the biggest index $b = \mb{ind}(m, m')$, if $b =
    n+1$, then $CT_j$ and $CT_{j'}$ are the ciphertexts of the same
    message. In the simulation, $\mc{B}$ finds the tuple $(-, (c'_{1,0},
    c'_{1,1}), \ldots, (c'_{n,0}, c'_{n,1}))$ from the table $T_{CT}$ and
    uses it to simulate the ciphertext $CT_i$ by computing $(C_{k,0},
    C_{k,1}) = ( {c'_{k,0}}^{s_j}, {c'_{k,1}}^{s_j} )$ for all $k \in [n]$.
    Let $CK_{j,j'} = (K_0, K_1)$ and we have
    \begin{align*}
    e(C_{k,0}, K_0) = e({c_{k,0}}^{s_j}, K_0) = e({c'_{k,0}}^{s_j}, K_0)
                    = e({c'_{k,0}}^{s_{j'}}, K_1) = e(C'_{k,0}, K_1) ~~\forall k\in [n].
    \end{align*}
    Otherwise, $m$ and $m'$ may have the same prefix of the length $b-1$.
    For $k \in [b-1]$, $(C_{k,0}, C_{k,1})$ has been simulated as the
    previous case and for $k \in [b+1, n]$, $(C_{k,0}, C_{k,1})$ has been
    simulated by using random elements. For the remain part $(C_{b,0},
    C_{b,1})$, $\mc{B}$ simulates $c_{b,1} = c'_{b,0}$ if $\mb{cmp}(m, m')
    = 1$. Then we have
    \begin{align*}
    e(C_{b,1}, K_0) = e({c_{b,1}}^{s_j}, K_0) = e({c'_{b,0}}^{s_j}, K_0)
                    = e({c'_{b,0}}^{s_{j'}}, K_1) = e(C'_{b,0}, K_1).
    \end{align*}
    Since we assumed that $CT_{i'}$ and $\overline{CT}_{i'}$ are
    identically distributed, by induction, $CT_i$ and $\overline{CT}_i$ are
    identically distributed.
\end{itemize}

In addition, suppose that the tuple $(E_0(k, m), h)$ is in $T_H$ and $(i,
(c_{1,0}, c_{1,1}),\ldots,(c_{n,0}, c_{n,1}))$ is in $T_{CT}$ for some $i$
such that $m_i = m$. By the Lemmas \ref{lem:mc-ore-1-h0-h1} and
\ref{lem:mc-ore-1-h1-h2}, $\mc{A}$ can not find out that $h$ and $c_{k,0}$
are different. This completes the correctness of simulation.
\end{proof}

\subsection{Extensions}

We present several extensions of our basic MC-ORE scheme to overcome their
shortcomings.

\vs\noindent \tb{Reducing Trust on the Center.} The basic MC-ORE scheme has
the problem that a center should be fully trusted because it generates the
secret keys of individual clients and comparison keys of different clients.
The existence of a trusted center is very strong constraint and it is costly
to ensure the security of such a center in reality. One way to reduce trust
on the center is that each client himself selects a secret key and securely
transfers the corresponding information to the center instead of having the
center owns the secret keys. That is, each client chooses its secret key
$s_j$ and securely sends $\hat{g}^{s_j}$ to the center, and then the center
can generate a comparison key $CK = ( (\hat{g}^{s_j})^r, (\hat{g}^{s_k})^r )$
by using $\hat{g}^{s_j}, \hat{g}^{s_k}$ received from clients and a random
exponent $r$. In this case, the center only can generate comparison keys, but
it can not generate client's ciphertexts since it does not have the secret
keys of individual clients.

\vs\noindent \tb{Removing the Trusted Center.} Unlike the previous ORE
schemes, our basic MC-ORE scheme requires a center to generate secret keys of
individual clients and comparison keys between different clients. Although we
suggested a method to reduce trust on the center, we cannot remove the
ability of the center to generate comparison keys. Note that if a comparison
key is exposed, a malicious client can compare any ciphertexts between two
clients by using the exposed comparison key. One idea to securely generate a
comparison key even after the center is completely removed is that two
clients perform a cryptographic protocol to share the same random value
$\hat{g}^r$ which is used to create $\hat{g}^{r s_j}$ and $\hat{g}^{r s_k}$.
The simplest way to non-interactively share the random value is to use a hash
function. That is, two clients with indices $j$ and $k$ generate $H(j \|
k)^{s_j}$ and $H(j \| k)^{s_k}$ respectively, and transmit these values to a
third client. Note that these values are a valid comparison key since $H(j \|
k)$ corresponds to $\hat{g}^r$ for some random exponent $r$.

\section{Enhanced MC-ORE Construction}

In this section, we propose our second construction of MC-ORE with reduced
leakage and prove the ST-SIM security of our scheme.

\subsection{Construction}

In the basic MC-ORE scheme, both ciphertext comparisons in a single client
and between different clients leak the most significant differing bit as well
as the result of the comparison. Although there are some ORE schemes with
reduced leakage \cite{LewiW16,CashLOZ16}, it is difficult to extend those
schemes to support comparisons on ciphertexts generated by different clients.
To build an MC-ORE scheme with reduced leakage, we divide the ciphertext into
independent two parts such that the first part only supports ciphertext
comparisons in a single client, and the second part only supports ciphertext
comparisons between different clients. For the first part, we use any ORE
scheme with reduced leakage. For the second part, we construct an encrypted
ORE (EORE) scheme by modifying our basic MC-ORE scheme so that it can not be
used for ciphertext comparisons in a single client. If the second part has no
leakage until a comparison key is provided, only the reduced leakage of the
ORE scheme affects the overall leakage.

\vs\noindent \tb{Encrypted ORE.} We first construct an EORE scheme by
modifying our basic MC-ORE scheme. The syntax of EORE is very similar to that
of MC-ORE defined in Definition \ref{def:mc-ore} except that the comparison
algorithm is excluded. The ciphertext of the EORE scheme is created by first
generating a ciphertext of the basic MC-ORE scheme and then encrypting it
with a public-key encryption scheme. The comparison key of the EORE scheme
includes additional elements that decrypt the encrypted ciphertext to obtain
the comparison form of the basic MC-ORE scheme. The ciphertext comparison is
performed in a similar manner to the basic MC-ORE scheme.

Let $S = \{(j,k)\}_{j,k \in [N]}$ be a set of index tuples where the
comparison key $CK_{j,k}$ is revealed. A leakage function $\mc{L}_S^{EORE}$
is defined as follows:
    \begin{align*}
    \mc{L}_S^{EORE} \big( (j_1, m_1), \cdots, (j_q, m_q) \big)
        = \big\{ \tb{cmp}(m_{i'}, m_{i}), \tb{ind}(m_{i'}, m_{i}):
          1 \le i' < i \le q, (j_{i'}, j_{i}) \in S \big\}.
    \end{align*}
Our EORE scheme with leakage $\mc{L}_S^{EORE}$ is given as follows:

\begin{description}
\item [\tb{EORE.Setup}($1^\lambda, N$).] This algorithm first generates
    bilinear groups $\G, \hat{\G}, \G_T$ of prime order $p$ with group
    generators $g \in \G$ and $\hat{g} \in \hat{\G}$. It chooses random
    exponents $s_j, a_j \in \Z_p$ and computes $h_j = g^{a_j}$ and
    $\hat{h}_j = \hat{g}^{a_j}$ for all $j\in[N]$. It outputs a master
    key $MK = \big( \{ s_j, \hat{h}_j \}_{j\in [N]} \big)$ and public
    parameters $PP = \big( (p, \G, \hat{\G}, \G_T, e), g, \hat{g},
    \{h_j\}_{j\in [N]}, H \big)$ where $H:\bits^* \to \G$ is a
    full-domain hash function.

\item [\tb{EORE.GenKey}($j, MK, PP$).] Let $MK = (\{ s_1, \cdots, s_N
    \}, \{\hat{h}_1, \cdots, \hat{h}_N \})$. It outputs a secret key
    $SK_j = s_j$.

\item [\tb{EORE.Encrypt}($m, SK_j, PP$).] Let $m = x_1 x_2 \cdots x_n \in
    \bits^n$ and $SK_j = s_j$.
    For each $i \in [n]$, it computes $F_{i,0} = H( \mb{prefix}(m, i) \| 0
    x_i )^{s_j}$ and $F_{i,1} = H( \mb{prefix}(m, i) \| (0 x_i + 1)
    )^{s_j}$. For each $F_{i,b}$, it selects a random exponent $t \in \Z_p$
    and computes $C_{i,b,0} = F_{i,b} h_j^t$ and $C_{i,b,1} = g^t$. It
    outputs a ciphertext $CT_j = \big( \{ C_{i,b,0}, C_{i,b,1} \}_{i \in
    [n], b \in \bits} \big)$.

\item [\tb{EORE.GenCmpKey}($j, k, MK, PP$).] Let $s_j$ and $s_k$ be the
    secret keys of client indices $j$ and $k$.
    It chooses a random exponent $r \in \Z_p$ and computes $K_{0,0} =
    \hat{g}^{r s_j}, K_{0,1} = \hat{h}_k^{r s_j}$ and $K_{1,0} = \hat{g}^{r
    s_k}, K_{1,1} = \hat{h}_j^{r s_k}$. It outputs the comparison key
    $CK_{j,k} = ( \{ K_{b,0}, K_{b,1} \}_{b \in \bits} )$.

\item [\tb{EORE.CompareMC}($CT_j, CT'_k, CK_{j,k}, PP$).] Let $CT_j = ( \{
    C_{i,b,0}, C_{i,b,1} \} )$ and $CT'_k = ( \{ C'_{i,b,0}, C'_{i,b,1} \}
    )$ for $i \in [n]$ and $b \in \bits$. Let $CK_{j,k} = ( \{ K_{b,0},
    K_{b,1} \}_{b\in\bits} )$.
    It first finds the smallest index $i^*$ such that $$e(C_{i^*,0,0},
    K_{1,0}) / e(C_{i^*,0,1}, K_{1,1}) \neq e(C'_{i^*,0,0}, K_{0,0}) /
    e(C'_{i^*,0,1}, K_{0,1})$$ by sequentially comparing these values from
    an index $0$ to $n$.
    If such index $i^*$ exists and $e(C_{i^*,1,0}, K_{1,0}) \lb /
    e(C_{i^*,1,1}, K_{1,1}) = e(C'_{i^*,0,0}, K_{0,0}) / e(C'_{i^*,0,1},
    K_{0,1})$ holds, then it outputs 1. If such index $i^*$ exists and
    $e(C_{i^*,0,0}, K_{1,0}) / e(C_{i^*,0,1}, K_{1,1}) = e(C'_{i^*,1,0},
    K_{0,0}) / e(C'_{i^*,1,1}, K_{0,1})$, then it outputs 0. If no such
    index $i^*$ exists, then it outputs 0.
\end{description}

\begin{remark}\label{remark:leakage-eore}
The leakage function $\mc{L}_S^{EORE}$ is same to the leakage function
$\mc{L}_S$ of the basic MC-ORE scheme except that it excludes the condition
$j_{i'} = j_i$. It means that the basic MC-ORE scheme leaks the comparison
result between ciphertexts of a single client, but the EORE scheme does not
leak any information before the comparison key is revealed.
\end{remark}

\noindent \tb{Multi-Client ORE}. Now we construct an enhanced MC-ORE scheme
by composing any ORE scheme with reduced leakage and the above EORE scheme.
As mentioned before, the ciphertext of the enhanced MC-ORE scheme consists of
two parts such that the first part is created from the ORE scheme and the
second part is created from the EORE scheme.

Let $\mc{L}_j^{ORE}$ be the leakage function of the underlying ORE scheme
corresponding to the client index $j$ and $\mc{L}_S^{EORE}$ be the leakage
function of our EORE scheme. A leakage function $\mc{L}_S^{MC\text{-}ORE}$ is
defined as follows:
    \begin{align*}
    \mc{L}_S^{MC\text{-}ORE} \big( (j_1, m_1), \cdots, (j_q, m_q) \big)
        = \big\{
        \mc{L}_j^{ORE}(m_{i_1}, \cdots, m_{i_\rho}) \cup \mc{L}_S^{EORE}:
        j = j_{i_1} = \cdots = j_{i_\rho}
        \big\}.
    \end{align*}
where the sequence sets $M_j = \{m_{i_1}, \cdots, m_{i_\rho}\} $ satisfy
$\bigcap M_j = \emptyset$ and $\bigcup M_j = \{m_1, \cdots, m_q\}$. Here, if
$S = \emptyset$, meaning that any comparison key is not revealed, then
$\mc{L}_S^{MC\text{-}ORE}$ becomes equal to the reduced leakage functions
$\{\mc{L}_j^{ORE}\}$ for each $j$. Otherwise, if $S\neq \emptyset$, to
achieve reducing the leakage, the ORE scheme is restricted from having no
leakage beyond the leakage of the EORE scheme for the same client. That is,
$\mc{L}_S^{MC\text{-}ORE}$ will be at most $\mc{L}_S^{EORE}$. Our MC-ORE
scheme with leakage $\mc{L}_S^{MC\text{-}ORE}$ that combines an ORE scheme
and our EORE scheme is described as follows:

\begin{description}
\item [\tb{MC-ORE.Setup}($1^\lambda, N$).] It obtains $MK_{EORE}$ and
    $PP_{EORE}$ by running $\tb{EORE.Setup} (1^{\lambda}, N)$ and
    outputs $MK = MK_{EORE}$ and $PP = PP_{EORE}$.

\item [\tb{MC-ORE.GenKey}($j, MK, PP$).] It runs $\tb{ORE.Setup}
    (1^{\lambda})$ and $\tb{EORE.GenKey} (j, MK, PP)$ to obtain
    $SK_{ORE,j}$ and $SK_{EORE,j}$, respectively. It outputs a secret key
    $SK_j = ( SK_{ORE,j}, SK_{EORE,j} )$.

\item [\tb{MC-ORE.Encrypt}($m, SK_j, PP$).] Let $SK_j = ( SK_{ORE,j},
    SK_{EORE,j} )$. It first obtains $OC_j$ and $EC_j$ by running
    $\tb{ORE.Encrypt} (m, SK_{ORE,j})$ and $\tb{EORE.Encrypt} (m,
    SK_{EORE,j}, PP)$ respectively. It outputs a ciphertext $CT_j = ( OC_j,
    EC_j )$.

\item [\tb{MC-ORE.Compare}($CT_j, CT'_j, PP$).] Let $CT_j = ( OC_j, EC_j )$
    and $CT'_j = ( OC'_j, EC'_j )$ for the same client index $j$. It
    returns $\tb{ORE.Compare}(OC_j, OC'_j)$.

\item [\tb{MC-ORE.GenCmpKey}($j, k, MK, PP$).] Let $SK_j$ and $SK_k$ be the
    secret keys for the client indices $j$ and $k$. It outputs the
    comparison key $CK_{j,k}$ by running $\tb{EORE.GenCmpKey}(j, k, MK,
    PP)$.

\item [\tb{MC-ORE.CompareMC}($CT_j, CT'_k, CK_{j,k}, PP$).] Let $CT_j = (
    OC_j, EC_j )$ and $CT'_k = ( OC'_k, EC'_k )$. It returns the result of
    $\tb{EORE.CompareMC} (EC_j, EC'_k, CK_{j,k}, PP)$.
\end{description}

\subsection{Correctness}

For the ciphertext comparisons in a single client, the correctness follows
from that of the underlying ORE scheme. For the ciphertext comparisons
between different clients, the correctness is shown as follows. Let $SK_j =
s_j$ and $SK_k = s_k$ be the secret keys of client indices $j$ and $k$, and
$CK_{j,k} = (K_{0,0}, K_{0,1}, K_{1,0}, K_{1,1}) = (\hat{g}^{r s_j},
\hat{h}_k^{r s_j}, \hat{g}^{r s_k}, \hat{h}_j^{r s_k})$ be the comparison key
of $(j,k)$. Let $EC_j = \big( \{ C_{i,b,0}, C_{i,b,1} \}_{i \in [n], b \in
\{0,1\}} \big)$ and $EC'_k = \big( \{ C'_{i,b,0}, C'_{i,b,1} \}_{i \in [n], b
\in \{0,1\}} \big)$ be ciphertexts on messages $m$ and $m'$. If $m < m'$,
there must be a smallest index $i^*$ such that $x_i = x'_i$ for all $i < i^*$
and $x_{i^*} \neq x'_{i^*}$. Then we have that
    \begin{align*}
    e(C_{i,0,0}, K_{1,0}) / e(C_{i,0,1}, K_{1,1})
        &= e(H(E_0(i,m))^{s_j} h_j^t, \hat{g}^{r s_k}) / e(g^t, \hat{h}_j^{r s_k})
        = e(H(E_0(i,m)), \hat{g})^{r s_j s_k} \\
        &= e(H(E_0(i,m'))^{s_k} h_k^{t'}, \hat{g}^{r s_j}) / e(g^{t'}, \hat{h}_k^{r s_j})
        = e(C'_{i,0,0}, K_{0,0}) / e(C'_{i,0,1}, K_{0,1}) ~~\forall i < i^*,\\
    e(C_{i^*,1,0}, K_{1,0}) / e(C_{i^*,1,1}, K_{1,1})
        &= e(H(E_1(i^*,m))^{s_j} h_j^t, \hat{g}^{r s_k}) / e(g^t, \hat{h}_j^{r s_k})
        = e(H(E_1(i^*,m)), \hat{g})^{r s_j s_k}\\
        &= e(H(E_0(i^*,m'))^{s_k} h_k^{t'}, \hat{g}^{r s_j}) / e(g^{t'}, \hat{h}_k^{r s_j})
        = e(C'_{i^*,0,0}, K_{0,0}) / e(C'_{i^*,0,1}, K_{0,1}).
    \end{align*}

\subsection{Security Analysis}

We now prove the security of the enhanced MC-ORE scheme with the leakage
function $\mc{L}_S^{MC\text{-}ORE}$ in the ST-SIM security model. We begin by
giving a high-level overview of the security proof. We define a sequence of
experiments from $\mb{H}_0$ corresponding to the real experiment to
$\mb{H}_4$ corresponding to the ideal experiment and show that the outputs of
two experiments are indistinguishable. Since the ciphertext is divided into
two parts: the ORE ciphertext $OC$, and the EORE ciphertext $EC$, the hybrid
experiments are also defined separately. At first, the ORE ciphertexts are
simulated only with the leakage functions $\mc{L}_j^{ORE}$. In the next
experiment, the EORE ciphertexts of clients whose comparison keys are not
exposed are randomly generated. Then, in the next experiment, the EORE
ciphertexts of clients whose comparison keys are exposed are generated with
random values. Finally, in the last experiment $\mb{H}_4$, the EORE
ciphertexts of clients whose comparison keys are exposed are simulated with
respect to the leakage function $\mc{L}_S^{EORE}$, and consequently
$\mb{H}_4$ corresponds to the ideal experiment. The details are given as
follows.

\begin{theorem} \label{thm:mc-ore-2}
The enhanced MC-ORE scheme is ST-SIM secure with the leakage function
$\mc{L}_S^{MC\text{-}ORE}$ in the random oracle model if the ORE scheme is
SIM secure with the leakage function $\mc{L}^{ORE}$, the basic MC-ORE scheme
is ST-SIM secure with the leakage function $\mc{L}_S$, and the XDH assumption
holds.
\end{theorem}

\begin{proof}
We prove the security of our enhanced MC-ORE scheme through a sequence of
hybrid experiments. The first experiment is defined as the real MC-ORE
security experiment and the last one is defined as the ideal experiment with
the leakage function $\mc{L}_S^{MC\text{-}ORE}$ in which the adversary has no
advantage. The hybrid experiments $\mb{H}_0, \mb{H}_1, \mb{H}_2, \mb{H}_3$,
and $\mb{H}_4$ are defined as follows:
\begin{description}
\item [$\mb{H}_0$]: This experiment corresponds to the real world
    experiment.

\item [$\mb{H}_1$]: In this experiment, the ORE ciphertexts $OC_j$ are
    generated with the leakage function $\mc{L}_j^{ORE}$ and the rest are
    same to $\mb{H}_0$. We have that $\mb{H}_0$ and $\mb{H}_1$ are
    indistinguishable if the underlying ORE scheme is secure with respect
    to the leakage function $\mc{L}^{ORE}$.

\item [$\mb{H}_2$]: This experiment is similar to $\mb{H}_1$ except that
    the EORE ciphertext $EC_j$ such that $(j, j') \not\in S$ for any client
    index $j'$ is generated by using random elements.

\item [$\mb{H}_3$]: This experiment is similar to $\mb{H}_2$ except that
    the EORE ciphertext $EC_j$ such that $(j, j') \in S$ for some client
    indices $j'$ is generated by using random elements.

\item [$\mb{H}_4$]: In this experiment, the EORE ciphertext $EC_{j_i}$ such
    that $(j_i, j)\in S$ for some client indices $j$ is generated with the
    leakage function $\mc{L}_S^{EORE}$ and the rest are same to $\mb{H}_3$.
    This experiment corresponds to the ideal world experiment.
\end{description}

From the following Lemmas \ref{lem:mc-ore-2-h0-h1}, \ref{lem:mc-ore-2-h1-h2},
\ref{lem:mc-ore-2-h2-h3}, and \ref{lem:mc-ore-2-h3-h4} that claim the
indistinguishability of the experiments, we have that $\mb{H}_0$ and
$\mb{H}_4$ are computationally indistinguishable.
\end{proof}

\begin{lemma} \label{lem:mc-ore-2-h0-h1}
The hybrid experiments $\mb{H}_0$ and $\mb{H}_1$ are computationally
indistinguishable to the polynomial-time adversary if the underlying ORE
scheme is SIM secure with the leakage function $\mc{L}^{ORE}$.
\end{lemma}

\begin{proof}
The proof of this lemma is simple since a ciphertext $CT_j$ consists of two
independent part $OC_j$ and $EC_j$. A simulator can use the simulator of the
ORE scheme for the generation of $OC_j$ and it can generate other elements in
$EC_j$ by the randomly chosen master key of an EORE scheme.
\end{proof}

\begin{lemma} \label{lem:mc-ore-2-h1-h2}
The hybrid experiments $\mb{H}_1$ and $\mb{H}_2$ are computationally
indistinguishable to the polynomial-time adversary assuming that the mXDH
assumption holds.
\end{lemma}

\begin{proof}
To prove this lemma, we define a sequence of hybrid experiments $\mb{H}_1 =
\mb{H}_{1,0}, \mb{H}_{1,1}, \cdots, \mb{H}_{1,q} = \mb{H}_2$ as follows.
\begin{description}
\item [$\mb{H}_{1, \mu}$]: In this experiment, we change the generation of
    the $\mu$-th ciphertext if $(j_\mu, *)\notin S$. If $i \le \mu$ and
    $(j_i, *)\notin S$, the $i$-th EORE ciphertext $EC_{j_i}$ is generated
    by using random elements. Otherwise, the $i$-th EORE ciphertext
    $EC_{j_i}$ is generated by running the normal encryption algorithm.
    Note that $\mb{H}_{1,\mu-1}$ and $\mb{H}_{1,\mu}$ are trivially equal
    if $(j_i, *)\in S$.
\end{description}

Without loss of generality, we assume that $(j_\mu, *) \notin S$. Suppose
there exists an adversary $\mc{A}$ that distinguishes $\mb{H}_{1, \mu-1}$
from $\mb{H}_{1, \mu}$ with non-negligible advantage. A simulator $\mc{B}$
that solves the mXDH assumption using $\mc{A}$ is given: a challenge tuple $D
= \big( (p, \G, \hat{\G}, \G_T, e), g, \hat{g}, g^a, g^{b_1}, \cdots,
g^{b_{2n}}) \big)$ and $T = ( X_1, \cdots, X_{2n} )$. $\mc{B}$ interacts with
$\mc{A}$ as follows.

Let $( st_{\mc{A}}, S, ((j_1, m_1), \cdots, (j_q, m_q)) )$ be the output of
$\mc{A}$. The simulator $\mc{B}$ first sets the public parameters
corresponding to the client index. For each $j\in [N]$, if $j \neq j_\mu$,
$\mc{S}$ chooses a random exponent $\alpha_j \in \Z_p$ and computes $h_j =
g^{\alpha_j}$. For the target client $j_\mu$, it sets $h_{j_\mu} = g^a$.
Next, for each $j \in [N]$, $\mc{B}$ chooses a random exponent $s_j \in \Z_p$
and sets the secret key $SK_j = s_j$. It can generate the comparison key
$CK_{j,k}$ for any tuple $(j,k) \in S$ since it knows secret keys $s_j$ and
$s_k$.

Let $(j_i, m_i)$ be the $i$-th ciphertext query for a client index $j_i$. Let
$E_0(k,m) = \mb{prefix}(m,k) \| 0x_k$ and $E_1(k,m) = \mb{prefix}(m,k) \|
(0x_k + 1)$ be encoded messages where $m = x_1 \cdots x_n \in \bits^n$. If
$(j_i, *) \in S$, then $\mc{B}$ simply creates a ciphertext by running the
\tb{EORE.Encrypt} algorithm since it know the secret key $s_{j_i}$. If $(j_i,
*) \notin S$, then $\mc{B}$ creates the $i$-th EORE ciphertext $EC_{j_i}$ as
follows:

\begin{itemize}
\item \tb{Case} $i < \mu$: It chooses random elements $R_{k,0} =
    (R_{k,0,0}, R_{k,1,0}), R_{k,1} = (R_{k,0,1}, R_{k,1,1})\in \G \times
    \G$ for all $k\in[n]$ and creates $EC_{j_i} = ( \{R_{k,0},
    R_{k,1}\}_{k\in [n]} )$.

\item \tb{Case} $i = \mu$: For each $\beta \in \bits$, it computes $F_{k,
    \beta} = H( E_\beta(k, m_i) )^{s_{j_i}}$ for all $k \in [n]$. For each
    $F_{k,0}$ and $F_{k,1}$, it sets $C_{k, 0} = (F_{k, 0}\cdot X_{2k-1},
    g^{b_{2k-1}})$ and $C_{k, 1} = (F_{k, 1}\cdot X_{2k}, g^{b_{2k}})$ and
    creates $EC_{j_i} = ( \{ C_{k,0}, C_{k,1} \}_{k\in [n]})$.

\item \tb{Case} $i > \mu$: It creates the EORE ciphertext $EC_{j_i}$ by
    running the \tb{EORE.Encrypt} algorithm.
\end{itemize}
If $T = (g^{ab_1}, \cdots, g^{ab_{2n}})$, then $EC_{j_\mu}$ is a ciphertext
in $\mb{H}_{1,\mu-1}$. Otherwise, $EC_{j_\mu}$ is a ciphertext in
$\mb{H}_{1,\mu}$. By the mXDH assumption, two experiments $\mb{H}_{1,\mu-1}$
and $\mb{H}_{1,\mu}$ are computationally indistinguishable.
\end{proof}

\begin{lemma} \label{lem:mc-ore-2-h2-h3}
The hybrid experiments $\mb{H}_2$ and $\mb{H}_3$ are computationally
indistinguishable to the polynomial-time adversary assuming that the mXDH
assumption holds.
\end{lemma}

\begin{proof}
We additionally define a sequence of hybrid experiments $\mb{H}_2 =
\mb{H}_{2, 0}, \mb{H}_{2, 1}, \cdots, \mb{H}_{2, \tilde{q}} = \mb{H}_3$ for
some $\tilde{q}$ as follows.

\begin{description}
\item [$\mb{H}_{2, \mu}$]: Let $I = ( j_1, \cdots, j_q )$ be a tuple of
    challenge client index and $\mb{RI}_\mu = \{ i \in [q] : j_i\text{s are
    co-related indices} \}$ be an index set of co-related client indices
    where $\mu \in [\tilde{q}]$. In this experiment, we change the
    generation of the $\mu$-th EORE ciphertext set with the index set
    $\mb{RI}_\mu$. If $\ell \le \mu$, the EORE ciphertexts in the $\ell$-th
    ciphertext set with $\mb{RI}_\ell$ are changed to be random elements.
    Otherwise, the ciphertexts in the $\ell$-th ciphertext set with
    $\mb{RI}_\ell$ are generated by running the normal encryption
    algorithm. Note that the ciphertexts with the client index $j_i$ such
    that $(j_i, *)\notin S$ in $\mb{H}_{2,\mu-1}$ and $\mb{H}_{2,\mu}$ are
    equally generated by using random elements.
\end{description}

Suppose there exists an adversary $\mc{A}$ that distinguishes $\mb{H}_{2,
\mu-1}$ from $\mb{H}_{2, \mu}$ with non-negligible advantage. A simulator
$\mc{B}$ that solves the mXDH assumption using $\mc{A}$ is given: a challenge
tuple $D = \big( (p, \G, \hat{\G}, \G_T, e), g,\lb \hat{g}, g^a, \{ g^{b_{i,1}},
\cdots, g^{b_{i,2n}} \}_{i \in [t]} \big)$ and $T = \big( \{ X_{i,1}, \cdots,
X_{i,2n} \}_{i \in [t]} \big)$. $\mc{B}$ runs the simulator
$\mc{B}_{bMC\text{-}ORE}$ of the Lemma \ref{lem:mc-ore-1-h1-h2} as a
subsimulator by submitting the challenge tuple of the mXDH assumption. Then
$\mc{B}$ that interacts with $\mc{A}$ is described as follows.

Let $( st_{\mc{A}}, S, ((j_1, m_1), \cdots, (j_q, m_q)) )$ be the output of
$\mc{A}$. The simulator $\mc{B}$ first sets the public parameters
corresponding to the client index. For each $j\in [N]$, $\mc{S}$ chooses a
random exponent $\alpha_j \in \Z_p$ and computes $h_j = g^{\alpha_j}$. For
each tuple $(j, k )\in S$, $\mc{B}$ obtains $CK'_{j,k} = ( K_0, K_1 )$ by
running $\mc{B}_{bMC\text{-}ORE}$ and computes the comparison key $CK_{j,k} =
( K_0, {K_0}^{\alpha_k}, K_1, K_1^{\alpha_j} )$.

For the creation of the EORE ciphertexts with the client index $j_i$ for $i
\in \mb{RI}_\ell$, $\mc{B}$ first runs $\mc{B}_{bMC\text{-}ORE}$ and obtains
$CT'_{j_i} = ( \{ F_{k,0}, F_{k,1} \}_{k \in [n]} )$. For each $F_{k,b}$, it
chooses a random exponent $t\in \Z_p$ and computes $C_{k,b} = (F_{k,b}\cdot
{h_{j_i}}^t, g^t)$ where $b\in \bits$. It creates $EC_{j_i} = ( \{ C_{k,0},
C_{k,1} \}_{k \in [n]} )$ and hence creates the EORE ciphertext set
$\mb{EC}_{RI_\ell} = ( \{ EC_{j_i} \}_{i \in RI_\ell} )$

By the Lemma \ref{lem:mc-ore-1-h1-h2}, two experiments $\mb{H}_{2,\mu-1}$ and
$\mb{H}_{2,\mu}$ are computationally indistinguishable.
\end{proof}

\begin{lemma} \label{lem:mc-ore-2-h3-h4}
The hybrid experiments $\mb{H}_3$ and $\mb{H}_4$ are indistinguishable to the
polynomial-time adversary with the leakage function $\mc{L}_S^{EORE}$ in the
random oracle model.
\end{lemma}

\begin{proof}
Suppose there exists an adversary $\mc{A}$ that distinguishes $\mb{H}_3$ from
$\mb{H}_4$ with non-negligible advantage. We construct an efficient simulator
$\mc{B}$ for which the two distributions $\mb{H}_3$ and $\mb{H}_4$ are
statistically indistinguishable. $\mc{B}$ runs the simulator
$\mc{B}_{bMC\text{-}ORE}$ of the Lemma \ref{lem:mc-ore-1-h2-h3} as a
subsimulator.

Let $( st_{\mc{A}}, S, ((j_1, m_1), \cdots, (j_q, m_q)) )$ be the output of
$\mc{A}$. The simulator $\mc{B}$ first sets the public parameters
corresponding to the client index. For each $j \in [N]$, $\mc{B}$ chooses a
random exponent $\alpha_j \in \Z_p$ and computes $h_j = g^{\alpha_j}$. For
each tuple $(j, k )\in S$, $\mc{B}$ obtains $CK'_{j,k} = ( K_0, K_1 )$ by
running $\mc{B}_{bMC\text{-}ORE}$ and computes the comparison key $CK_{j,k} =
( K_0, {K_0}^{\alpha_k}, K_1, K_1^{\alpha_j} )$.

For the generation of the $i$-th EORE ciphertext $EC_{j_i}$ with the client
index $j_i$, $\mc{B}$ first runs $\mc{B}_{bMC\text{-}ORE}$ and obtains
$CT'_{j_i} = ( \{ F_{k,0}, F_{k,1} \}_{k \in [n]} )$. For each $F_{k,b}$, it
chooses a random exponent $t\in \Z_p$ and computes $C_{k,b} = (F_{k,b}\cdot
{h_{j_i}}^t, g^t)$ where $b\in \bits$. It creates $EC_{j_i} = ( \{ C_{k,0},
C_{k,1} \}_{k \in [n]} )$.

By the Lemma \ref{lem:mc-ore-1-h2-h3}, the distributions $(CT'_{j_1}, \ldots,
CT'_{j_q})$ and $(\overline{CT'}_{j_1}, \ldots, \overline{CT'}_{j_q})$ of the
basic MC-ORE ciphertexts output in $\mb{H}_2$ and $\mb{H}_3$ of the Theorem
\ref{thm:mc-ore-1} are indistinguishable. Thus, it can be easily derived that
the distributions $(EC_{j_1}, \ldots, EC_{j_q})$ and $(\overline{EC}_{j_1},
\ldots, \overline{EC}_{j_q})$ in $\mb{H}_3$ and $\mb{H}_4$ are also
indistinguishable. Since $OC_j$ and $EC_j$ have been generated with respect
to $\mc{L}_j^{ORE}$ and $\mc{L}_S^{EORE}$, $\mb{H}_4$ corresponds to the
ideal experiment.
\end{proof}

\section{Implementation}

In this section, we measure the performance of our MC-ORE schemes and compare
various ciphertext comparison methods. Our implementation is entirely written
in C and employs a 224-bit MNT curves from the PBC library for pairing
operations. We run our implementation on a laptop with 4GHz Intel Core
i7-6700K CPU and 16GB RAM.

\subsection{Performance of MC-ORE}

We evaluate the runtime of \tb{Encrypt}, \tb{Compare}, and \tb{CompareMC}
algorithms for 32-bit integers and the benchmarks averaged over 100
iterations are given in Table \ref{tab:mcore-bench}. Compared to the basic
MC-ORE scheme, the encrypted ORE scheme takes more time to run each algorithm
and the size of the ciphertext and the comparison key is about twice as
large. The reason why the encrypted ORE scheme is less efficient is that its
\tb{Encrypt} algorithm runs the \tb{Encrypt} algorithm of the basic MC-ORE
scheme as a subalgorithm and the \tb{CompareMC} algorithm requires twice as
many pairing operations as the basic MC-ORE scheme. This shows that although
the security of the MC-ORE scheme is improved by reducing the leakage, at the
same time, the efficiency is decreased. We note that the runtime of the
\tb{Compare} algorithm and the accurate size of the ciphertext of the
enhanced MC-ORE scheme depend on the underlying ORE scheme.

\begin{table*}[t]
\caption{Performance comparison between our MC-ORE schemes} \label{tab:mcore-bench}
\centering
\vs \small \addtolength{\tabcolsep}{5.5pt}
\renewcommand{\arraystretch}{1.2}
\newcommand{\otoprule}{\midrule[0.09em]}
    \begin{tabularx}{6.10in}{lrrrrr}
    \toprule
    Scheme          & Encrypt $(ms)$    & Compare $(\mu s)$    & CompareMC $(ms)$   &
    $|CT|$          & $|CK|$ \\
    \otoprule
    Basic MC-ORE    & 45.8     & 1.65   & 35.6    &$2n |\G|$     & $2 |\hat{\G}|$ \\
    Encrypted ORE   & 107.4    & -      & 58.2    &$> 4n |\G|$   & $4 |\hat{\G}|$ \\
    \bottomrule
    \end{tabularx}
\end{table*}

\subsection{Range Query Methods}

One possible application of MC-ORE is a range query for an encrypted
database, in which case a database server must perform the multi-client
comparison algorithm many times to find a subset of database that satisfies
the range query. Suppose that a database sever keeps each client database
$D_j \in [R]^M$ that store ciphertexts generated by a client with index $j$
where the database consists of maximum $M$ values each in the range $[R]$. A
client with an index $k$ may request a range query by giving $CT'$ on a
plaintext $m'$ encrypted with $SK_k$ to find a subset of ciphertexts in $D_j$
less than $m'$. If the server has a comparison key $CK_{j,k}$, then it can
answer the query by simply running the multi-client comparison algorithm $M$
times. However, this naive method is very slow since $M$ multi-client
comparison operations are required. Therefore, we need better methods to
handle range queries by using comparison operations more efficiently.

We present two additional methods and compare these methods with the simple
method described before. The detailed explanation of each method is given as
follows.
\begin{itemize}
\item \tb{Simple Method.} The simple method simply runs the \tb{CompareMC}
    algorithm $M$ times. Recall that the \tb{CompareMC} algorithm tries to
    find the MSDB from the higher bit to the lower bit sequentially. If the
    MSDB is located in higher bits, then the comparison operation is
    considerably efficient. On the other hand (if the MSDB is located in
    lower bits), the comparison operation is relatively slow.

\item \tb{BinSearch Method.} The binary search method uses a modification
    of the \tb{CompareMC} algorithm that finds the MSDB more efficiently by
    using binary searching instead of sequential searching. Let $CT =
    \{C_{i,0}, C_{i,1}\}_{i \in [n]}$ be one ciphertext in a database $D_j$
    and $CT' = \{C'_{i,0}, C'_{i,1}\}_{i \in [n]}$ be a ciphertext created
    by a client with $k$. A server with a comparison key $CK_{j,k} = (K_0,
    K_1)$ first checks whether $e(C_{n/2,0}, K_1)$ and $e(C'_{n/2,0}, K_0)$
    are equal or not. If the values are equal, it checks again whether
    $e(C_{3n/4, 0}, K_1)$ and $e(C'_{3n/4,0}, K_0)$ are equal since the
    MSDB is in the range $[n/2 +1,n]$. On the other hand, if the values are
    not equal, it checks whether $e(C_{n/4,0}, K_1)$ and $e(C'_{n/4,0},
    K_0)$ are equal since the MSDB is in the range $[1,n/2]$. By repeating
    this process $\log n$ times, the server can find the MSDB. Since the
    database contains at most $M$ entries, it runs this modified comparison
    algorithm $M$ times.

\item \tb{Hybrid Method.} The hybrid method uses the \tb{CompareMC}
    algorithm and the \tb{Compare} algorithm together since the
    \tb{Compare} algorithm is fast and it can compare the order of
    ciphertexts in a database which are created by a single client. Let
    $CT_i$ be a ciphertext on a message $m_i$ in a database $D_j$ and $CT'$
    be a ciphertext on a message $m'$ given by a client in a range query.
    To answer the range query of the client, a server should find a subset
    $S$ of ciphertexts in $D_j$ such that $m_i < m'$. The server first
    compares $CT'$ with one specific $CT_{i^*} \in D_j$ by running the
    \tb{CompareMC} algorithm, and then it divides all other $CT_i \in D_j$
    into two groups $L$ and $R$ by running the \tb{Compare} algorithm on
    input $CT_i$ and $CT_{i^*}$ where $L$ contains ciphertexts of $m_i <
    m_{i^*}$ and $R$ contains ciphertexts of $m_{i^*} \leq m_i$. If
    $m_{i^*} < m'$, then the server adds $L$ to the subset $S$ and repeats
    the above process by setting $D_j = R$. If $m' < m_{i^*}$, the sever
    repeasts the above process by setting $D_j = L$.
\end{itemize}

\begin{table*}[t]
\caption{Performance comparison between range query methods} \label{tab:method-comp}
\centering
\vs \small \addtolength{\tabcolsep}{3.5pt}
\renewcommand{\arraystretch}{1.2}
\newcommand{\otoprule}{\midrule[0.09em]}
    \begin{tabularx}{4.90in}{lccc}
    \toprule
    $R$    & Simple Method $(sec)$ & BinSearch Method $(sec)$ & Hybrid Method $(sec)$ \\
    \otoprule
    $2^8$       & 11.50    & 4.79      & 2.46     \\
    $2^{16}$    & 8.19     & 4.84      & 1.73     \\
    $2^{24}$    & 5.90     & 4.89      & 1.27     \\
    $2^{28}$    & 4.55     & 4.89      & 1.10     \\
    $2^{32}$    & 3.48     & 4.82      & 0.87     \\
    \bottomrule
    \end{tabularx}
\end{table*}

We compared the performance of each method only for the basic MC-ORE scheme.
For the comparison, we set $M = 100$ and $R \in \{ 2^8, 2^{16}, 2^{24},
2^{28}, 2^{32} \}$. We randomly selected 32-bit integers $m_1, \cdots,
m_{100}$ and $m'$ within a specific range $[0, R]$, and then we encrypted
$m_1, \cdots, m_{100}$ with $SK$ and $m'$ with $SK'$. The running time of the
above three range query methods is given in Table \ref{tab:method-comp}.
The binary search method executes 12 pairing operations per a single
ciphertext comparison whereas the simple method performs executes at least 4
up to 66 pairing operations depending on the location of the MSDB. In other
words, the binary search method is better than the simple method if the data
are within a small range and the high-order bits are equal, but it is less
efficient if the data are within a large range and the probability that the
high-order bits are equal is lower. The hybrid method is always more
efficient than the simple method and the binary search method, since some
comparisons are performed by using the \tb{Compare} algorithm instead of
using the \tb{CompareMC} algorithm. That is, the performance of the hybrid
method is far superior because the \tb{CompareMC} algorithm is performed for
comparisons on the specific ciphertexts and the \tb{Compare} algorithm is
executed for the remaining ciphertext comparisons.

It is important to improve the performance of the algorithm, but our results
show that efficiency can be improved in an appropriate way depending on the
application environment. If our MC-ORE scheme is applied to an environment
other than a database range query, we can consider another way to improve its
performance or its security.

\section{Conclusion}

We introduced the concept of multi-client order-revealing encryption (MC-ORE)
that supports comparisons on ciphertexts generated by multiple clients as
well as generated by one client. We also defined the simulation-based
security model for MC-ORE with respect to a leakage function. We then
proposed two practical MC-ORE schemes with different leakage functions and
proved their security in the defined security model. The first scheme leaks
more information, namely the most significant differing bit, and the second
scheme is the enhanced scheme with reduced leakage. We implemented our
schemes to measure the performance of each algorithm and provided additional
range query methods to improve the performance in a database range query.

\bibliographystyle{plain}
\bibliography{mc-ore-with-leakage}

\end{document}